\documentclass[onefignum,onetabnum]{siamonline171218}

\usepackage{lipsum}
\usepackage{amsfonts}
\usepackage{graphicx}
\usepackage{epstopdf}
\ifpdf
  \DeclareGraphicsExtensions{.eps,.pdf,.png,.jpg}
\else
  \DeclareGraphicsExtensions{.eps}
\fi

\newsiamremark{remark}{Remark}
\newsiamremark{hypothesis}{Hypothesis}
\crefname{hypothesis}{Hypothesis}{Hypotheses}
\newsiamthm{claim}{Claim}

\headers{Beltrami Optimization for Spherical Bijection}{Z.H.Xu, L.M.Lui}

\title{Two-chart Beltrami Optimization for Distortion-Controlled Spherical Bijection with Application to Brain Surface Registration\thanks{Submitted to the editors DATE.
\funding{The second author was supported by HKRGC GRF (Project ID: 14307621).
}}}
\author{Zhehao Xu\thanks{Department of Mathematics, The Chinese University of Hong Kong
  (\email{zhxu@math.cuhk.edu.hk}
  ).}
\and Lok Ming Lui\thanks{Department of Mathematics, The Chinese University of Hong Kong
  (\email{lmlui@math.cuhk.edu.hk}
  ).}
}

\usepackage{amsopn}

\makeatletter
\newcommand*{\addFileDependency}[1]{
  \typeout{(#1)}
  \@addtofilelist{#1}
  \IfFileExists{#1}{}{\typeout{No file #1.}}
}
\makeatother



\newcommand*{\myexternaldocument}[1]{%
    \externaldocument{#1}%
    \addFileDependency{#1.tex}%
    \addFileDependency{#1.aux}%
}

\usepackage{booktabs}
\usepackage{lipsum}
\usepackage{amsfonts}
\usepackage{graphicx}
\usepackage{epstopdf}
\usepackage{enumerate}
\usepackage{enumitem}
\usepackage{cite}
\usepackage{subcaption}

\usepackage{tabularx}
\usepackage{arydshln}
\usepackage{multirow}
\usepackage{caption}
\usepackage{microtype}

\usepackage{algorithm}
\usepackage{algpseudocode}
\usepackage{amsmath}
\allowdisplaybreaks[1]

\usepackage{mathtools}

\usepackage{svg}

\ifpdf
\hypersetup{
  pdftitle={Spherical Beltrami Differential},
  pdfauthor={Z.H. Xu, and L.M. Lui}
}
\fi

\myexternaldocument{ex_supplement}

\begin{document}

\maketitle

\begin{abstract}
Many genus-0 surface mapping tasks such as landmark alignment, feature matching, and image-driven registration, can be reduced (via an initial spherical conformal map) to optimizing a spherical self-homeomorphism with controlled distortion. However, existing works lack efficient mechanisms to control the geometric distortion of the resulting mapping. To resolve this issue, we formulate this as a Beltrami-space optimization problem, where the angle distortion is encoded explicitly by the Beltrami differential and bijectivity can be enforced through the constraint $\|\mu\|_{\infty}<1$. To make this practical on the sphere, we introduce the Spherical Beltrami Differential (SBD), a two-chart representation of quasiconformal self-maps of the unit sphere $\mathbb{S}^2$, together with cross-chart consistency conditions that yield a globally bijective spherical deformation (up to conformal automorphisms). Building on the Spectral Beltrami Network, we develop BOOST, a differentiable optimization framework that updates two Beltrami fields to minimize task-driven losses while regularizing distortion and enforcing consistency along the seam. Experiments on large-deformation landmark matching and intensity-based spherical registration demonstrate improved task performance meanwhile maintaining controlled distortion and robust bijective behavior. We also apply the method to cortical surface registration by aligning sulcal landmarks and matching cortical sulcal depth, achieving comparative or better registration performance without sacrificing geometric validity.

\end{abstract}



\section{Introduction}
Surface registration refers to the process of establishing a meaningful bijective correspondence between two geometric surfaces. This fundamental task serves as a prerequisite for comparing shapes, transferring attributes, and analyzing geometric variations across different objects. Applications can be found in many fields, ranging from computer vision \cite{chen2024three,li2024bi} to medical imaging and computational anatomy \cite{chen2024multiscale,ren2024sugar}. In particular, registration between genus-0 closed surfaces has attracted significant attention in recent years due to its critical role in analyzing biological structures with spherical topology.  For example, brain surface registration enables standardized comparisons between individuals for statistical analysis \cite{lui2007landmark,ren2024sugar,zhao2022fast}. Another related application is the computation of spherical parameterizations with controlled distortion, which facilitate geometry processing and numerical analysis on genus‑0 closed surfaces \cite{gu2004genus,choi2015flash,wang2016bijective,liao2024convergence}. 

The registration can be done by parameterizing the surfaces to a common template space, specifically the spherical surface. The problem of optimizing a mapping from a genus-0 closed surface to another one is then boiled down to finding an optimal sphere self-mapping. This is achieved by computing the respective spherical parameterization of each surface, whose existence is guaranteed by the uniformization theorem. Numerically, there are many established tools \cite{choi2015flash, gu2004genus,jin2008discrete,liao2024convergence} that enable its efficient computation. Therefore, let $\phi_i:\mathcal{M}_i\to \mathbb{S}^2$ be the spherical parameterization of each $\mathcal{M}_i$, then the problem can be formulated mathematically as finding an optimal sphere self-mapping such that
\begin{equation}
     f = \textbf{argmin}_{g : \mathbb{S}^{2} \rightarrow \mathbb{S}^{2}} E(g) \label{eq: formulation of optimization problem} \\ \text{ subject to } g \in \mathcal{C},
\end{equation}
and the composition $\phi_2^{-1}\circ f \circ \phi_1$ is the desired mapping from $\mathcal{M}_1$ to $\mathcal{M}_2$. 
In this formulation, $E$ is the energy functional by which the solution satisfies desired properties, and $\mathcal{C}$ is the constraint space.  In applications, bijectivity is essential for the map. Beyond bijectivity, we require the map to have \emph{well-controlled geometric distortion} so that it is physically plausible. 
Accordingly, the feasible set $\mathcal{C}$ should be understood as bijective maps together with an explicit distortion constraint or regularization (e.g., bounds on local anisotropy), rather than arbitrary smooth bijections.
For example, in the landmark constrained conformal mapping problem\cite{lui2010optimized,lam2014landmark}, to get a bijective and plausible deformation, it is necessary to control the local geometric distortion like angle distortion.

Broadly, existing approaches either optimize vertex positions directly on the sphere (e.g., energy minimization\cite{lai2014folding} or flow-based updates\cite{yeo2009spherical}) or pull the problem back to a planar domain via stereographic projection and optimize in $\mathbb{C}$ \cite{choi2016spherical, choi2015flash, haker2002conformal, liao2024convergence}. While intuitive and widely used, direct optimization \cite{lai2014folding, yeo2009spherical,zhao2022fast, zhao2019spherical} typically leads to a highly nonconvex landscape and does not inherently enforce global bijectivity.
Chart-based formulations \cite{choi2016spherical, choi2015flash, haker2002conformal, liao2024convergence, guo2023automatic} alleviate some numerical difficulty but introduce their own geometric constraints. They usually introduce artificial boundary/pole handling because the sphere cannot be covered by one chart and this can induce significant distortion near the pole. In addition, current approaches, when applied to registration problems, often struggle to simultaneously achieve both bijectivity and low geometric distortion. A more detailed review is given in Section \ref{sec: previous works}.

These limitations motivate our exploration of representing spherical bijection through distortion variables rather than vertex coordinates and developing optimization frameworks on the distortion variables directly so as to control the geometric distortion of the result mapping efficiently. In quasiconformal theory, the Beltrami differential/coefficient provides a compact, intrinsic description of local shear and stretch and there are mature algorithms for reconstructing quasiconformal mappings with Beltrami coefficient given \cite{lui2012optimization,meng2016tempo,lui2013texture}. However, two key obstacles remain: (i) on the sphere, a quasiconformal representation and its computation naturally require a multi-chart viewpoint and a free-boundary mapping algorithm, and (ii) classical solvers are not designed as differentiable primitives for modern gradient-based optimization. In this work, we introduce the spherical Beltrami differential as a two-chart representation with explicit cross-chart consistency, and combine it with a differentiable quasiconformal solver surrogate, the \textbf{S}pectral \textbf{B}eltrami \textbf{N}etwork (\textbf{SBN}), based on least-squares quasiconformal energy (LSQC) \cite{qiu2019computing,xu2025neural} to enable efficient, distortion-controlled optimization for spherical mapping tasks. 
In particular, to solve the bijective spherical parameterization problem, we propose a novel algorithm, \textbf{B}eltrami \textbf{O}ptimization \textbf{O}n \textbf{S}pherical \textbf{T}opology (\textbf{BOOST}). BOOST represents a spherical map using two stereographic charts (hemispheres) and optimizes the corresponding Beltrami fields using the SBN. The spherical Beltrami differential has a correspondence with the spherical quasiconformal map once cross-chart compatibility is satisfied. In practice, BOOST enforces this compatibility through seam matching, folding penalties after stitching, and seam smoothness regularization. This two-chart formulation avoids the pole distortion inherent in single-chart spherical parameterizations and enables flexible, distortion-controlled non-overlap registration on genus-0 surfaces, with applications to cortical surface alignment and landmark/feature or hybrid matching.

The organization of the paper is as follows. 
In Section \ref{sec: previous works}, related works are reviewed, and contributions are highlighted in Section \ref{sec: contributions}.
In Section \ref{sec: math background}, necessary mathematical concepts are introduced. 
In Section \ref{sec: spherical beltrami differential}, we propose the concept of the spherical Beltrami differential and show its correspondence with spherical quasiconformal mappings, and then introduce the model Spectral Beltrami Network.
In Section \ref{sec: BOOST}, we establish the neural-based optimization framework BOOST.
In Section \ref{sec: experiments}, multiple experiments are conducted to show the capability and potential of our work in complex imaging problems, with applications to cortical surface registration.
In Section \ref{Sec: discussion}, we analyze the key factors hidden behind the excellence of the model and briefly discuss its limitations as well as potential future work.
In Section \ref{sec: conclusions}, a conclusion is given.

\section{Previous works}
\label{sec: previous works}

Spherical mapping for genus‑0 surfaces plays an important role in surface registration problems. A common pipeline is to first compute an initial spherical conformal map and then solve a spherical self-mapping problem to meet downstream goals such as distortion reduction, landmark alignment, feature matching, or image-driven registration. In this section, we review related work from three angles: (i) distortion-driven spherical parameterization, (ii) task-driven spherical registration/correspondence, and (iii) distortion control via quasiconformal/Beltrami representations.
\paragraph{Distortion-driven spherical parameterization}
A large body of work focuses on mapping genus--0 surfaces to the sphere with bounded or optimized distortion, primarily as a parameterization or remeshing tool. Early and influential approaches include (i) direct spherical parameterization such as \cite{praun2003spherical,wan2012efficient} based on nonlinear energies, flow based methods \cite{jin2008discrete, kazhdan2012can}, density-equalizing/area-preserving mapping \cite{lyu2024spherical,liu2026spherical,cui2019spherical} and  as-rigid-as-possible (ARAP) spherical parameterization \cite{wang2014rigid} that trades off angle and area distortions. (ii) parameterization domain based approaches like \cite{haker2002conformal, angenent1999conformal, angenent2002laplace} that compute conformal parameterization via solving PDE on the extended complex plane, and the subsequent improvement \cite{choi2015flash,choi2016spherical,liao2024convergence} that propose north-south pole iterative schemes to mitigate angle distortion around poles. These methods are highly effective when the goal itself is a low-distortion parameterization, but they do not directly address general mapping tasks that must simultaneously satisfy task objectives and maintain controllable distortion.

\paragraph{Task-driven spherical mapping and registration}
In many imaging applications, the spherical map is not the end of the story but a representation that converts surface matching to registering signals on the sphere (e.g., curvature, sulcal depth, functional features). In computational neuroanatomy, classical spherical registration is central to cortical surface analysis and atlas building. Representative pipelines include spherical coordinate systems and spherical morphing used in widely adopted tools \cite{fischl1999high}, diffeomorphic registration on the sphere such as Spherical Demons \cite{yeo2009spherical}, and flexible frameworks like Multimodal Surface Matching (MSM) that align multiple features with explicit smoothness/regularization models \cite{robinson2014msm}. Recent years have seen significant progress in learning-based spherical surface registration. Spherical CNNs and spherical U-Net-type models introduced practical spherical convolutions for tasks such as segmentation/parcellation and spherical signal analysis \cite{cohen2018spherical,zhao2019spherical}. Examples include SphereMorph \cite{cheng2020cortical}, S3Reg \cite{Zhao2021S3Reg}, the deep-discrete spherical registration framework (DDR) \cite{suliman2022deep}, and recent graph-based spherical registration models such as SUGAR \cite{ren2024sugar}. These methods substantially decrease runtime and often achieve strong task accuracy. However, in these methods, the topology is preserved mainly through folding penalties and/or diffeomorphic parameterizations of the deformation model, while explicit and tunable control of geometric distortion is usually absent or only via an indirect smoothness regularization on the result mappings. This limitation becomes particularly acute when the task loss encourages highly localized stretching or twisting.

\paragraph{Distortion control via quasiconformal theory}
Quasiconformal (QC) geometry offers a principled mechanism to quantify and constrain distortion: the Beltrami coefficient/differential encodes local shear and anisotropic stretch, and regularization on its norm can control conformality distortion. This has motivated a series of work that represents mappings in distortion variables rather than vertex coordinates, enabling more direct control of geometric validity. Examples include Teichm\"uller maps for landmark matching and surface registration \cite{lui2014teichmuller}, Beltrami holomorphic flow \cite{lui2012optimization} and related Beltrami-based numerical solvers \cite{lui2013texture, qiu2019computing, meng2016tempo}, and optimization over constrained homeomorphism spaces parameterized by Beltrami coefficients \cite{lam2014landmark,lui2015splitting}. In parallel, quasi-conformal theories' application to spherical topology has also been studied, often through chart-based constructions \cite{choi2015flash,choi2016spherical}. While these approaches bring explicit distortion measures into the loop, practical gaps remain when applying QC/Beltrami formulations to spherical topology: (i) the sphere cannot be covered by a single chart, so need to adopt an iterative scheme to address pole/boundary artifacts; (ii) the optimization is typically carried out with ADMM-style schemes rather than end-to-end differentiable updates, making it less convenient to integrate general task losses and to scale to modern gradient-based learning/optimization pipelines; (iii) existing QC pipelines are often formulated around explicit landmark or boundary constraints to anchor the solution, whereas many practical registration settings are landmark-free and instead rely on dense feature or intensity objectives.

Overall, prior literature provides powerful tools for (i) spherical parameterization with optimized distortion and (ii) task-driven spherical registration, but relatively few approaches can jointly optimize task objectives while enforcing bijectivity and providing explicit, effective distortion control. This gap motivates approaches that represent spherical self-maps in distortion variables (e.g., Beltrami fields), use multi-chart formulations to avoid pole artifacts, and support efficient gradient-based optimization for diverse task losses.

\section{Contributions}
\label{sec: contributions}

Our work has several key contributions and novel aspects to the field of spherical bijective mapping problems, addressing the limitations of existing methods. Specifically, our contributions are as follows:
\begin{enumerate}
\item \textbf{Two-chart Beltrami formulation on the sphere.}
We introduce the \emph{Spherical Beltrami Differential (SBD)}, a two-chart representation of spherical quasiconformal self-maps with cross-chart compatibility. This provides a practical parameterization of spherical self-homeomorphisms in terms of distortion variables.

\item \textbf{BOOST: task-driven optimization of spherical self-maps in Beltrami space.}
We propose \emph{BOOST}, a differentiable optimization framework that updates two Beltrami fields to minimize task losses while regularizing distortion and enforcing global consistency across the seam (via boundary matching, seam smoothness, and anti-folding penalties). The framework targets general spherical mapping/registration problems.

\item \textbf{Validation on spherical registration tasks.}
We demonstrate the effectiveness of the proposed formulation and optimization strategy on representative tasks, including landmark matching, intensity-based spherical registration and hybrid objectives, with experiments conducted on brain cortical surface data. Results show improved task alignment with better controlled distortion and consistent bijective behavior.
\end{enumerate}

\section{Mathematical Background}
\label{sec: math background}

\subsection{Quasiconformal Mappings}
Quasiconformal mappings generalize conformal mappings by extending orientation-preserving homeomorphisms to those with bounded conformality distortions.
Mathematically, suppose $\Omega$ is a domain in $\mathbb{C}$, a mapping $f \colon \Omega \rightarrow \mathbb{C}$ is a quasi-conformal map if it satisfies the Beltrami equation
\begin{equation}\label{eq:beltrami}
    \frac{\partial f}{\partial \bar{z}} = \mu(z)\frac{\partial f}{\partial z}
\end{equation}
for some measurable complex-valued function $\mu$ with $\| \mu \|_{\infty}<1$, where $\frac{\partial f}{\partial z} = \frac{1}{2}(f_{x}-\sqrt{-1}f_{y})$ and $\frac{\partial f}{\partial \bar{z}} = \frac{1}{2}(f_{x}+\sqrt{-1}f_{y})$. $\mu$ is called the complex dilation or Beltrami coefficient of $f$,  which measures the local deviation from a conformal map. Infinitesimally, around a point $p$, $f(p+z)$ might be expressed as an affine map with negligible error $o(|z|)$:
\begin{equation}
    f(p+z) \approx f(p) + f_{z}(p)z + f_{\bar{z}}(p)\bar{z} = f(p) + f_{z}(p)(z+\mu(p)\bar{z})
\end{equation}
for a conformal map, its Taylor expansion excludes $\bar{z}$-terms, implying $\mu = 0$, and thus $\mu$ is a measure of nonconformality. Locally, $f$ can be viewed as a composition of a translation to $f(p)$, a stretch map $S(z)=z+\mu(p)\bar{z}$, and a multiplication of $f_{z}(p)$. The stretch map $S(z)$ distorts a circle to an ellipse, with $\mu(p)$ determining the angles of maximal magnification ($\arg \mu /2$ with factor $1+|\mu|$) and maximal shrinkage  ($(\arg \mu + \pi) /2$ with factor $1-|\mu|$).

Let $f \colon (x,y) \rightarrow (u,v) $ and $\mu = \rho + \sqrt{-1} \tau$, and denote by $\mathbf{S}(2)$ the space of all $2 \times 2$ symmetric positive definite matrices whose determinant is 1, then 
\begin{equation}
    D_f(z)^{T}D_{f}(z) = |\det D_f(z)|Q(z) 
\end{equation}
where $Q= (Q_{ij}) \colon \Omega \rightarrow \mathbf{S}(2)$. Right multiplying both sides by $D_f(z)^{-1}$,  we obtain a linear system which is the alternative formulation of the Beltrami equation:
\begin{equation} \label{eq:4.4}
    \begin{bmatrix}
    u_x & u_y \\
    v_x & v_y
    \end{bmatrix}^T = \text{sgn}\left(J_f(x)\right) \cdot 
    \begin{bmatrix}
    q_{11} & q_{12} \\
    q_{12} & q_{22}
    \end{bmatrix}
    \begin{bmatrix}
    v_y & -u_y \\
    -v_x & u_x
    \end{bmatrix},
\end{equation}
and from this we have $\mu = \frac{q_{11}-q_{22}+2\sqrt{-1}q_{12}}{q_{11}+q_{22}+2\text{sgn}(\det{D_f})}$. With the constraint that $Q(z) \in \mathbf{S}(2)$, we observe that $|\mu(z)|<1$ if and only if $\det D_f(z)>0$. By inverse function theorem, researchers would through restricting $\| \mu \|_{\infty}<1$, achieve local bijectivity and prevent folding in the triangular mesh \cite{lui2014teichmuller, choi2015flash}.

Suppose $\mu(z)$ is a measurable complex-valued function defined in
a domain $U\subseteq \mathbb{C}$ for which $\| \mu \|_{\infty}<1$, we have the following existence theorem.
\begin{theorem}
\label{thm: Riemann}
    (\textbf{Measurable Riemann Mapping Theorem}) For any function $\mu : U \to \mathbb{C}$ on with bounded essential supremum norm $\|\mu\|_\infty < 1$, there is a quasiconformal map $\phi$ on $\overline{U}$ satisfying the Beltrami equation $\phi_{\overline{z}} = \mu \phi_z$ for almost all $z \in U$. Moreover, $\phi$ is unique up to post-composition with conformal isomorphisms and $\phi$ depends holomorphically on $\mu$.
\end{theorem}

For general Riemann surfaces, we have a similar concept, Beltrami differential.
\begin{definition}
    A Beltrami differential $\mu(z)\frac{\overline{dz}}{dz}$ on a Riemann surface $\mathcal{R}$ is an assignment to each chart $z_{\alpha}$ on $U_{\alpha}$ an $\mathcal{L}_{\infty}$ complex-valued function $\mu_{\alpha}$ defined on $z_{\alpha}(U_{\alpha})$ such that
    \[
    \mu_{\alpha}(z_{\alpha}) = \mu_{\beta}(z_{\beta}) \frac{\overline{(\frac{dz_{\beta}}{dz_{\alpha}})}}{\frac{dz_{\beta}}{dz_{\alpha}}}
    \]
    and $\| \mu \|_{\infty} = \sup_{\alpha}\|\mu_{\alpha}\|_{\infty}$
\end{definition}
The space of essentially bounded, complex-valued measurable Beltrami differentials on $\mathcal{R}$ is a Banach space, denoted by $L_{\infty}(\mathcal{R})$. The open unit ball of $L_{\infty}(\mathcal{R})$ is denoted by $M(\mathcal{R})$. Elements $\mu = (\mu_{\alpha})_{\alpha}$ in $M(\mathcal{R})$ are also called Beltrami coefficients.

\subsection{Least Squares Quasi-conformal Energy}

We define the unweighted LSQC energy as the $L^2$-norm of the Beltrami equation mismatch.
\begin{definition}
    \label{def: lsqc energy}
    Let \( \mu = \rho + \sqrt{-1}\tau \) be a complex-valued function defined on the domain \( \Omega \). The LSQC energy of the map \( z = (x,y) \mapsto (u,v) \) against the BC \( \mu \) is defined to be 
    $$\tilde{E}_{LSQC}(f, \mu) = \int_{\Omega} |f_{\bar{z}} - \mu f_z|^2 dA.$$ 
    Equivalently, 
    \begin{equation}
    \label{eq:4.7}
        \tilde{E}_{LSQC}(u,v, \mu) = \frac{1}{2} \int_{\Omega} \|\tilde{P}\nabla u + J\tilde{P}\nabla v\|^2 \, dx \, dy,
    \end{equation}
    where 
    $\tilde{P}=
    \begin{bmatrix}
    1 - \rho & -\tau \\ 
    -\tau & 1 + \rho
    \end{bmatrix}.$
\end{definition}

Given a piecewise linear map $f(x,y) = u+\sqrt{-1}v $ and consider a triangle $T=(x_j,y_j)_{j=1,2,3}$ of $\mathbb{R}^2$, we have 
$$
\begin{pmatrix}
    \partial u / \partial x \\
    \partial u / \partial y
\end{pmatrix}= 
\frac{1}{d_T}
\begin{pmatrix}
    y_2 - y_3 & y_3-y_1 & y_1-y_2\\
    x_3-x_2 & x_3-x_1 & x_1-x_2
\end{pmatrix}
\begin{pmatrix}
    u_1\\u_2\\u_3
\end{pmatrix}
$$
where $d_T=(x_1y_2-y_1x_2)+(x_2y_3-y_2x_3)+(x_3y_1-y_3x_1)$ is twice the area of the triangle. The Beltrami equation (Equation \ref{eq:beltrami}) can be rewritten as $(1-\mu)f_x + \sqrt{-1}(1+\mu)f_y=0$ and collected as a complex-valued linear system.
$$
0 = \frac{i}{d_T}(W_1  \quad W_2 \quad W_3)(U_1 \quad U_2 \quad U_3)^\top
$$
where $U_j=u_j+\sqrt{-1}v_j$ and 
\begin{align*}
    W_1 & = (1+\mu)(x_3-x_2)+\sqrt{-1}(1-\mu)(y_3-y_2) \\
    W_2 & = (1+\mu)(x_1-x_3)+\sqrt{-1}(1-\mu)(y_1-y_3) \\
    W_3 & = (1+\mu)(x_2-x_1)+\sqrt{-1}(1-\mu)(y_2-y_1).
\end{align*}
As a result, 
\begin{align*}
    E_{LSQC}(\textbf{U}=(U_1,\ldots,U_{|\mathcal{V}|})) &= \sum\limits_{T \in \mathcal{T}}d_T \Big|\left( (1-\mu)f_x + \sqrt{-1}(1+\mu)f_y \right)\Bigg|_T \Big|^{2} \\
    &=\sum_{T\in \mathcal{T}}\frac{1}{d_T} \Big|
    (W_{1,T}  \quad W_{2,T} \quad W_{3,T})(U_{1,T} \quad U_{2,T} \quad U_{3,T})^\top\Big|^2
\end{align*}
and we may write as $E_{LSQC}(\textbf{U})=\| \mathcal{M}\textbf{U}\|^2$ with $\mathcal{M}=(m_{ij}) \in \mathbb{C}^{|\mathcal{F}| \times |\mathcal{V}|}$ defined as
\[
\mathcal{M}_{ij} = 
\begin{cases} 
w_{ij} = W_{j,T_i} & \text{if j is a vertex of the face } T_i, \\ 
0 & \text{otherwise}.
\end{cases}.
\]
To obtain a nontrivial solution, some of the $U_i$'s have to be pinned. 

We decomposed $\mathbf{U}=(\mathbf{U}_f^\top, \mathbf{U}_{p}^\top)^\top$, where $\mathbf{U}_f$ are free points, i.e. variables of the optimization problem, and $\mathbf{U}_f$ are the points pinned. Similarly, we can decompose $\mathcal{M}$ in blocked matrices
\[
\mathcal{M} = \begin{pmatrix}
    \mathcal{M}_f, \mathcal{M}_p
\end{pmatrix}
\]
with $\mathcal{M}_f \in \mathbb{C}^{|\mathcal{F}| \times (|\mathcal{V}|-p)} $ and $\mathcal{M}_p \in \mathbb{C}^{|\mathcal{F}| \times p} $.
Let $^1$ and $^2$ be real and imaginary parts of a complex number, and rewrite the $|\mathcal{F}|$ complex linear equations to $2|\mathcal{F}|$ real equations, we have 
\[
E_{LSQC}(\mathbf{U}) = \| \mathcal{A} \mathbf{u}-\mathbf{b} \|^2
\]
where $\mathbf{u}=((\mathbf{U}_{f}^{1}) ^\top,(\mathbf{U}^2_f)^\top)^\top$ and 
\[
\mathcal{A} = \begin{pmatrix}
    \mathcal{M}_f^1 & -\mathcal{M}_f^2 \\
    \mathcal{M}_f^2 & \mathcal{M}_f^1
\end{pmatrix}, 
\mathbf{b} = - \begin{pmatrix}
    \mathcal{M}_p^1 & -\mathcal{M}_p^2 \\
    \mathcal{M}_p^2 & \mathcal{M}_p^1
\end{pmatrix}
\begin{pmatrix}
    \mathbf{U}_p^1\\
    \mathbf{U}_p^2
\end{pmatrix}.
\]
Finally, the nontrivial solution is $\mathbf{U}=\begin{pmatrix}
    (\mathcal{A}^\top \mathcal{A})^{-1}\mathcal{A}^{\top}\mathbf{b}\\ \mathbf{U}_p
\end{pmatrix}$.

In the previous work \cite{xu2025neural}, we prove multiple key properties of LSQC, listed as follows:
\begin{proposition}(\textbf{Nontriviality})
    Suppose $|\mu|$ is uniformly bounded away from 1, and the triangulation mesh is connected without dangling triangles. As long as $p \geq 2$, then $\mathcal{A}$ has full rank.
\end{proposition}
\begin{proposition}(\textbf{Exact solution to Beltrami Equation})
    Suppose $|\mu|$ is uniformly bounded away from 1,  and $\mathbf{U}$ is the solution. If $|\mathcal{F}|=|\mathcal{V}|-2$, then $\frac{\partial \mathbf{U}}{\partial \bar{z}}=\mu \frac{\partial \mathbf{U}}{\partial z}$.
\end{proposition}
\begin{proposition}(\textbf{Invariance by similarity transformation})
\label{prop: independent of similarity transformation}
    Given $\mu$ and pinned points $\mathbf{U}_p$ and denote the corresponding free-point solution by $\mathbf{U}_{f}$, if the pinned points are transformed to $z\mathbf{U}_p+\mathbf{T}$ where $z \in\mathbb{C}, \mathbf{T} = (z',\ldots,z')^\top\in \mathbb{C}^{|\mathcal{V}|}$, then the corresponding free-point solution is $z\mathbf{U}_f+\mathbf{T}$.
\end{proposition}
\begin{proposition}(\textbf{Resolution independence})
\label{prop: resolution independence}
    Suppose we have a triangulation $\mathcal{T}=(\mathcal{V},\mathcal{F})$ with piecewise constant $\mu = \{\mu_{T}: T \in \mathcal{F}\}$, and pinned points $\mathbf{U}_p$, and denote the corresponding solution by $\mathbf{U} = (\mathbf{U}_f, \mathbf{U}_p)$. Let $T=(\mathbf{v}_1,\mathbf{v}_2,\mathbf{v}_3)$ be a face in $T$ and we split $T$ into three triangles with $\mathbf{v}=\sum_{i=1}^{3}\alpha_j\mathbf{v}_j~ (\alpha_j>0)$ being the introduced vertex and $T_i$ being the new triangles not containing $\mathbf{v}_i$ as vertex. If $\mu_{T_i}=\mu_{T}~(i=1,2,3)$ and BCs on other faces remain, then the minimizer of this new minimization problem is $\mathbf{U}^{+}=(\mathbf{U}^\top, \mathbf{U}_{v})^\top$ where $\mathbf{U}_v=\sum_{i=1}^{3}\alpha_i\mathbf{U}_i$.
\end{proposition}
For detailed proofs, see \cite{xu2025neural}.

\section{Spherical Beltrami Differential and Spectral Beltrami Network}
\label{sec: spherical beltrami differential}
In this section, we present our theoretical contributions that form the mathematical foundation for neural-based spherical mapping optimization. 
We first establish fundamental correspondence results that extend the classical Beltrami theory to the spherical setting and then propose a Beltrami differential optimization framework for the problem, followed by introducing the Spectral Beltrami Network.
\subsection{Spherical Beltrami Differential}
In this subsection, we introduce the concept of spherical Beltrami diffferential associated to a spherical quasiconformal map. Before establishing the relation between spherical Beltrami differentials and spherical quasiconformal maps, we give two lemmas.
\begin{lemma}\label{lemma:4.4}
    Let $f$ be a quasiconformal map from $\mathcal{R} = (U_{\alpha},z_{\alpha})_{\alpha}$ to $\mathcal{R'}=$ \linebreak $ (V_{\beta}, w_{\beta})_{\beta}$, then the Beltrami coefficients $\mu_{\alpha}$ of $w_{\beta} \circ f \circ (z_{\alpha})^{-1}$ determine an element $\mu$ in $M(\mathcal{R}).$
\end{lemma}

\begin{proof}
    Denote $w_{\beta} \circ f \circ (z_{\alpha})^{-1}$ by $f_{\alpha \beta}$. We first need to prove that the Beltrami coefficients $\mu_{\alpha, \beta}$ of $f_{\alpha \beta}$ are independent of the choice of parameterizations $w_{\beta}$, thus can be denoted as $\mu_{\alpha}$.

    Suppose $\beta_1, \beta_2$ are two charts such that $f(U_{\alpha}) \cap V_{\beta_1} \cap V_{\beta_2} \neq \emptyset$ and denote the intersection by $V_{\alpha,1,2}$. Then for $z \in z_{\alpha}^{-1} \circ f^{-1}(V_{\alpha,1,2})$, $\mu_{\alpha, \beta_{1}}(z) = (f_{\alpha, \beta_1})_{\bar{z}} / (f_{\alpha, \beta_1})_{z}$ and $\mu_{\alpha, \beta_{2}}(z) = (f_{\alpha, \beta_2})_{\bar{z}} / (f_{\alpha, \beta_2})_{z}$. If we denote the conformal transition $w_{\beta_2} \circ w_{\beta_1}^{-1}$ by $w_{\beta_1,\beta_2}$, then $f_{\alpha, \beta_2} = w_{\beta_1,\beta_2} \circ f_{\alpha, \beta_1} $ and by chain rule and $(w_{\beta_1,\beta_2})_{\bar{z}} = 0$ , we have
    \begin{equation*}
        \begin{aligned}
           \mu_{\alpha, \beta_{2}}(z) & = \frac{(w_{\beta_1,\beta_2})_{z}  \cdot (f_{\alpha, \beta_1})_{\bar{z}} + (w_{\beta_1,\beta_2})_{\bar{z}}  \cdot (\overline{f_{\alpha, \beta_1}})_{\bar{z}} }{(w_{\beta_1,\beta_2})_{z}  \cdot (f_{\alpha, \beta_1})_{z} + (w_{\beta_1,\beta_2})_{\bar{z}}  \cdot (\overline{f_{\alpha, \beta_1}})_{z}} \\
           &=\frac{(w_{\beta_1,\beta_2})_{z}  \cdot (f_{\alpha, \beta_1})_{\bar{z}}}{(w_{\beta_1,\beta_2})_{z}  \cdot (f_{\alpha, \beta_1})_{z}} = \mu_{\alpha, \beta_{1}}(z)
        \end{aligned}
    \end{equation*}

    Next, show that the family of Beltrami coefficients induced by $f$ satisfies the consistency requirement of Beltrami differential. Denote the conformal transition $z_{\alpha_2} \circ z_{\alpha_1}^{-1}$ by $z_{\alpha_1,\alpha_2}$, then $f_{\alpha_1, \beta} = f_{\alpha_2, \beta} \circ z_{\alpha_1,\alpha_2}$ and we have
    \begin{equation}
    \label{eq:transition map beltrami differential}
        \begin{aligned}
           \mu_{\alpha_1}(z) & = \frac{( f_{\alpha_2, \beta} \circ z_{\alpha_1,\alpha_2})_{\bar{z}}}{( f_{\alpha_2, \beta} \circ z_{\alpha_1,\alpha_2})_{z}} = \frac{( f_{\alpha_2, \beta})_z \cdot (z_{\alpha_1,\alpha_2})_{\bar{z}} + ( f_{\alpha_2, \beta})_{\bar{z}} \cdot (\overline{z_{\alpha_1,\alpha_2}})_{\bar{z}}}
           {( f_{\alpha_2, \beta})_z \cdot (z_{\alpha_1,\alpha_2})_{z} + ( f_{\alpha_2, \beta})_{\bar{z}} \cdot (\overline{z_{\alpha_1,\alpha_2}})_{z}} \\
           &= \frac{( f_{\alpha_2, \beta})_{\bar{z}} \cdot (\overline{z_{\alpha_1,\alpha_2}})_{\bar{z}}}{( f_{\alpha_2, \beta})_z \cdot (z_{\alpha_1,\alpha_2})_{z} + ( f_{\alpha_2, \beta})_{\bar{z}} \cdot (\overline{z_{\alpha_1,\alpha_2}})_{z}} \\
           &= 
           \frac{( f_{\alpha_2, \beta})_{\bar{z}} \cdot \overline{(z_{\alpha_1,\alpha_2})_z}}{( f_{\alpha_2, \beta})_z \cdot (z_{\alpha_1,\alpha_2})_{z} + ( f_{\alpha_2, \beta})_{\bar{z}} \cdot \overline{(z_{\alpha_1,\alpha_2})_{\bar{z}}}} \\
           &= \mu_{\alpha_2}(z) \frac{\overline{(\frac{dz_{\alpha_2}}{dz_{\alpha_1}})}}{\frac{dz_{\alpha_2}}{dz_{\alpha_1}}}  
        \end{aligned}
    \end{equation}
    
    Therefore, a quasiconformal map would naturally determine a Beltrami differential on the domain surface, which is invariant under the parameterization of image.

\end{proof}

\begin{lemma} \label{lemma:identity map}
    Let $\mathcal{R} = (U_{\alpha},z_{\alpha})_{\alpha}$ be a Riemann surface and $\mu \in M(\mathcal{R})$, there exists a new atlas of charts $(U_{\alpha},w_{\alpha})_{\alpha}$ of the topological space $\mathcal{R}$ such that the identity map $\text{id} : \mathcal{R} = (U_{\alpha},z_{\alpha})_{\alpha} \rightarrow \mathcal{R}^{\mu} = (U_{\alpha},w_{\alpha})_{\alpha}$ is quasiconformal with Beltrami coefficients $\mu$.
\end{lemma}

\begin{proof}
    By Theorem \ref{thm: Riemann}, there is a quasiconformal homeomorphism $f = f^{\alpha}$ defined on each $z_{\alpha}(U_{\alpha})$ such that $f_{\bar{z}}(z) = \mu_{\alpha}(z)f_{z}(z)$. On each open set $U_{\alpha}$ of the given covering of $\mathcal{R}$, we introduce the new chart $w_{\alpha} = f^{\alpha} \circ z_{\alpha}$.

    First, the transition maps for this new atlas satisfy
    \[
    f_{\alpha \beta} = w_{\alpha} \circ (w_{\beta})^{-1} = f^{\alpha} \circ ( f^{\beta} \circ (z_{\beta} \circ z_{\alpha}^{-1})) ^{-1}
    \]
    by Equation \ref{eq:transition map beltrami differential}, consider $f^{\alpha}$ and $f^{\beta} \circ (z_{\beta} \circ z_{\alpha} ^{-1})$ 
    on the overlapping regions in the $z_{\alpha}(U_{\alpha})$, both are the solutions of the same Beltrami equations (with coefficients $\mu_{\alpha}$. By Theorem \ref{thm: Riemann}, there exists a conformal map $\varphi$ such that $f_{\alpha} = \varphi \circ (f^{\beta} \circ (z_{\beta} \circ z_{\alpha} ^{-1}))$, and thus $f_{\alpha \beta} = \varphi$ is conformal.

    Second, within this new system of charts, the identity map expressed in local coordinates is $w_{\alpha} \circ id \circ z_{\alpha}^{-1} = f^{\alpha}$ would satisfy the Beltrami equation with coefficients $\mu^{\alpha}$. 
\end{proof}

With these, we have the following existence theorem.
\begin{theorem}
\label{thm: beltrami differential correspondence}
    Let $\mathcal{R}$ and $\mathcal{R'}$ be two conformal structures of the unit sphere $\mathbb{S}^2$, then for any $\mu \in M(\mathcal{R})$, there exists a quasiconformal mapping $f : \mathcal{R} \rightarrow \mathcal{R'}$ satisfying the Beltrami equation, which is unique up to the postcomposition of conformal mappings.
\end{theorem}

\begin{proof}
    Given $\mu \in M(\mathcal{R})$, by Lemma \ref{lemma:identity map}, there exists a new chart $\mathcal{R}^{\mu}$ such that $id:\mathcal{R} \rightarrow \mathcal{R}^{\mu}$ is quasiconformal w.r.t $\mu$. By uniformization theorem, there exists a conformal mapping $f: \mathcal{R}^{\mu} \rightarrow \mathcal{\mathcal{R}'}$ and it is the desired quasiconformal mapping from $\mathcal{R}$ to $\mathcal{R}$'  
\end{proof}
By Uniformization Theorem, the result is still true in the more general case.
\begin{corollary}
    Let $\mathcal{R}$ and $\mathcal{R'}$ be two genus-0 closed surface, then for any $\mu \in M(\mathcal{R})$, there exists a quasiconformal mapping $f : \mathcal{R} \rightarrow \mathcal{R'}$ satisfying the Beltrami equation, which is unique up to the postcomposition of conformal mappings.
\end{corollary}

Motivated by this, we define a simplified variant of Beltrami differential on the sphere. 
\begin{sloppypar}
\begin{definition}
    A spherical Beltrami differential on $\mathbb{S}^2$ is a pair $\boldsymbol{\mu}_{\mathbb{S}^2} = \{(\mu_N, U_N, P_N), (\mu_S, U_S, P_S)\}$ such that
    \begin{itemize}
        \item $U_N = \mathbb{S}^2 \setminus V_N$ where $V_N \subseteq \mathbb{S}^2$ is a neighborhood of the north pole. $P_N: U_N \to \mathbb{C}$ is the stereographic projection from the north pole.
        \item $U_S = \mathbb{S}^2 \setminus V_S$ where $V_S \subseteq \mathbb{S}^2$ is a neighborhood of the south pole.  $P_S: U_S \to \mathbb{C}$ is the stereographic projection from the south pole.
        \item $U_N \cup U_S = \mathbb{S}^2$.
    
    \item  $\mu_N$ is a Beltrami coefficient defined on $P_N(U_N) \subset \mathbb{C}$ and
    $\mu_S$ is a Beltrami coefficient defined on $P_S(U_S) \subset \mathbb{C}$ such that $\max\{\|\mu_N\|_{\infty},\|\mu_S\|_{\infty}\} < 1$ and on the region $P_S(U_N \cap U_S)$, the coefficients satisfy:
   $$\mu_S(z) = \mu_N(1/z) \cdot (\frac{\bar{z}}{z})^{-2}$$
   where $z \mapsto 1/z$ is the transition function $P_S \circ P_N^{-1}$.
   \end{itemize}

\end{definition}
\end{sloppypar}
\begin{sloppypar}
    Using Theorem \ref{thm: beltrami differential correspondence}, we can establish a one-to-one correspondence between spherical self-quasiconformal maps and spherical Beltrami differentials, modulo the conformal automorphism group.
\end{sloppypar}

\begin{corollary}
\label{corollary: correspondence Theorem of spherical beltrami differential}
    Suppose $\mathcal{M} = \{(U_N, P_N), ( U_S, P_S)\}$ is a conformal structure of spherical surface, then for any spherical Beltrami differential $\boldsymbol{\mu}_{\mathbb{S}^2} = \{(\mu_N, U_N, P_N), (\mu_S, U_S, P_S)\}$, there exists a quasiconformal mapping $f : \mathcal{M} \rightarrow \mathbb{S}^2$ satisfying the Beltrami equation, which is unique up to the postcomposition of conformal mappings.
\end{corollary}
Our aim is to achieve spherical self-homeomorphism optimization with explicit control on the local geometric distortion. Considering the correspondence between spherical Beltrami differentials and spherical self-quasiconformal maps, we may formulate the problem as optimization over these geometric quantities. Through stereographic projection, the optimization problem for a spherical surface $\mathcal{M} = \{(U_N, P_N), ( U_S, P_S)\}$ (Equation \ref{eq: formulation of optimization problem}) can be reformulated as the optimization over spherical Beltrami differentials: 
\begin{equation}\label{eq:new formulation of the optimization problem, formulation 1}
\begin{aligned}
\min_{\mu = \mu_{\mathbb{S}^2}} \quad & E_1(f^{\mu}) + E_2(\mu)\\
\textrm{s.t.} \quad & f^{\mu}(\mathbb{S}^2)= P_{N}^{-1} \circ f^{N} \circ P_{N}(U_N) \cup P_{S}^{-1} \circ f^{S} \circ P_{S}(U_S)  \\& (f^{N})_{\bar{z}} = \mu_{N} (f^{N})_{z} \text{ , } (f^{S})_{\bar{z}} = \mu_{S} (f^{S})_{z}  \\
  & \mu_S(z) = \mu_N(1/z) \cdot (\frac{\bar{z}}{z})^{-2} , \quad z \in P_{S}(U_{N} \cap U_{S})  \\
\end{aligned}
\end{equation}

This new formulation includes two energy terms: $E_1$ drives task-specific objectives while $E_2$ regularizes BCs $\mu$ by constraining its magnitude (e.g. minimizing $|\mu|^2$) and smoothness (e.g. minimizing $|\nabla \mu|^2$). 

Due to Corollary \ref{corollary: correspondence Theorem of spherical beltrami differential}, the above formulation is equivalent to 
\begin{equation}\label{eq:new formulation of the optimization problem, formulation 2}
\begin{aligned}
\min_{\mu = \mu_{\mathbb{S}^2}} \quad & E_1(f^{\mu}) + E_2(\mu)\\
\textrm{s.t.} \quad & f^{\mu}:\mathbb{S}^2 \rightarrow \mathbb{S}^2 \text{ is a homeomorphism} \\
    & \mu_N = \mu(f^{\mu}\big|_{U_N}), \mu_S = \mu(f^{\mu}\big|_{U_S})
\end{aligned}
\end{equation}
where $\mu(f\big|_{U})$ is the Beltrami coefficient induced by the restriction of $f$ on $U$.
The Beltrami differential depends only on the parameter domain of the surface parameterization and for simplicity, in application we express the sphere self-mapping $f$ in Equation \ref{eq: formulation of optimization problem} as $f \colon \mathbb{S}^2 = (P_S^{-1}(\overline{\mathbb{D}}),P_S)\cup (P_N^{-1}(\overline{\mathbb{D}}),P_N) \rightarrow \mathbb{S}^2 =(\Omega_1, P_N)\cup (\Omega_2, P_S)$, where $P_S, P_N$ are south, north pole stereographic projections and $\Omega_j$ are simply connected domains in $\mathbb{R}^2$. From this point of view, we can decompose $f$ as the combination of its restriction on the upper and lower hemispheres, respectively. Equivalently, it can be viewed as the combination of two homeomorphisms on different hemispheres, and their images of the equator coincide. Therefore, the formulation is simplified to
\begin{equation}\label{eq:new formulation of the optimization problem, formulation 3}
\begin{aligned}
\min_{\mu_S,\mu_N : \mathbb{D} \rightarrow \mathbb{C}} \quad & E_1(f^{S},f^{N}) + E_2(\mu_S,\mu_N)\\
\textrm{s.t.} \quad & (f^{N})_{\bar{z}} = \mu_{N} (f^{N})_{z} \, , \, (f^{S})_{\bar{z}} = \mu_{S} (f^{S})_{z}\, , z\in \mathbb{D} \\
    & P_S^{-1}\circ f^S(z) = P_N^{-1}\circ f^N(\frac{1}{z})\, ,z\in\partial\mathbb{D}
\end{aligned}
\end{equation}

\subsection{Spectral Beltrami Network}
A key challenge is developing a differentiable tool to compute quasiconformal mappings from Beltrami differentials, enabling gradient-based optimization. Meanwhile, the mapping should be not subject to any boundary constraints because the shape of the target domain $f^{\mu_\alpha} \circ z_{\alpha}(U_{\alpha})$ is usually not prescribed. An important reason why LSQC is preferred in our work is that it implies that geometry of target domain is intrinsically determined by Beltrami coefficients. This characteristic aligns with the previously outlined Beltrami differential optimization formulation \ref{eq:new formulation of the optimization problem, formulation 3} and provides the flexibility to adjust boundary configurations dynamically through the optimization of the Beltrami differential.  

Following the notation in \cite{xu2025neural}, we can view the numerical LSQC algorithm as an operator $\mathcal{F}$. Let $(V,F)$ be a triangulation of the unit disk, $\mu=\{\mu_{T}:T\in F\}$ be piecewise constant Beltrami coefficients, $(p_{i},q_{i})_{i=1,2}$ be the two pairs of selected points in the mesh and their corresponding prescribed destinations in $\mathbb{R}^2$, then the numerical solution $f=\mathcal{F}(\mu, p_1,p_2, q_1,q_2)$ such that $f(p_i)=q_{i}$ and $(\frac{\partial f}{\partial \bar{z}} / \frac{\partial f}{\partial z})\Big|_{T} \approx \mu_{T}$.
For the aim of applying LSQC to gradient-based optimization problems, Spectral Beltrami Network (SBN) is proposed in the previous work\cite{xu2025neural}. Given Beltrami coefficients $\{\mu_v:v \in V^1\}$ defined for vertices on a prescribed mesh of $\mathbb{D}$ (or equivalently a directed graph) $\mathcal{G}^1 = (V^1,E^1,F^1)$ and coordinates of two pinned points $(p_1,p_2)$, SBN predicts the mapping $f(V^1)$ such that $f(V^1) \approx \mathcal{F}(\mu,p_1,p_2,p_1,p_2)$ where $\mu=\{\mu_{T}:\mu_{T}=\frac{1}{3}\sum_{v\in T}\mu_v, T \in F^1\}$. Arbitrary target positions $q_1,q_2$ can be recovered by post-composition with a similarity transformation $g$ satisfying $g(p_j)=q_{j}$ $(j=1,2)$. By Proposition \ref{prop: independent of similarity transformation}, $\mathcal{F}(\mu,p_1,p_2,q_1,q_2)=g\circ \mathcal{F}(\mu,p_1,p_2,p_1,p_2)$, ensuring no loss of generality. For capturing short-range and long-range information exchange between distinct vertices in the mesh, SBN couples mesh spectral layers with the multiscale message passing mechanism. The mesh spectral layer is to project the latent feature vectors into a span of low-frequency mesh Laplacian eigenvectors, and do the multi-node linear transformation in the frequency domain followed by an inverse projection back to the spatial domain, through which global shape context is injected to the latent space. In addition, the multiscale mesh hierarchy can widen the receptive field without exploding memory, improving the efficiency of information exchange through downsampling and upsampling. Figure \ref{fig: SBN} shows the architecture of SBN. For more details about the network, please refer to \cite{xu2025neural}.
The key contribution of SBN is that 1. through it, the Beltrami field becomes learnable and 2. with the stop-gradient operator, the choice of pinned points is no longer a discrete manner but also differentiable.
\begin{figure}[!htbp]
    \centering
    \includegraphics[width=\textwidth]{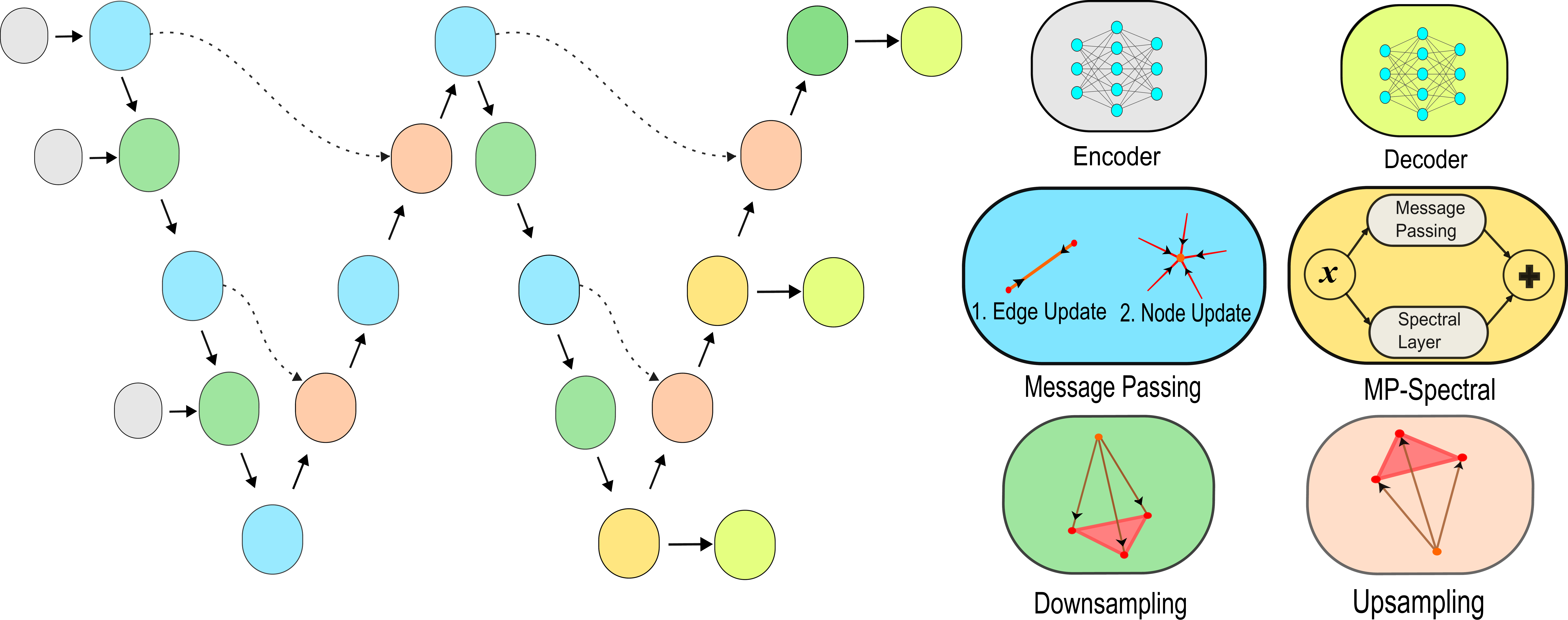}
    \caption{ \small \textit{
        Architecture of Spectral Beltrami Network}
    }
    \label{fig: SBN}
\end{figure}
\section{Beltrami Optimization On Spherical Topology (BOOST)}
\label{sec: BOOST}
The Spectral Beltrami Network (SBN) can predict free-boundary homeomorphisms on the unit disk $\mathbb{D}$. In the following, we introduce the optimization framework BOOST for the general spherical parameterization problem, which incorporates SBN for freely deforming the 2D parameterization domain.

After the training of SBN, we freeze its parameters. The well-trained SBN model serves as a tool for spherical surface parameterization, transforming a hemisphere into a submanifold of the unit sphere. A significant advantage of the SBN is that the pinned points $P = \{p_1, p_2\}$ are treated as trainable parameters rather than fixed constraints, using the stop gradient operator. By optimizing the Beltrami coefficients (BCs) and the coordinates of two fixed points, followed by post-composition with scaling, rotation, and translation, we derive optimal mappings associated to specific tasks.

To ensure valid inputs, the input BCs and fixed points must have norms less than 1. Let $\tilde{\mu}_v$, $T_{\text{BC}}$, $\tilde{p}_i$ and $T_{\text{pin}}$ be the parameters to be optimized, then define an activation function $\mathcal{T}$
\[
\mathcal{T}(x,T) = \left( \frac{e^{\frac{|x|}{T}} - e^{-\frac{|x|}{T}}}{e^{\frac{|x|}{T}} + e^{-\frac{|x|}{T}}} \right)e^{i\arg(x)}
\]
and $\mu_v = \mathcal{T}(\tilde{\mu}_v,T_{\text{BC}})$ and $p_i = \mathcal{T}(\tilde{p}_i, T_{\text{pin}})$ are the input to the model.

Now, we propose an optimization framework for spherical mapping problems, which jointly optimizes the mappings of both hemispheres via SBN,  \textbf{B}eltrami \textbf{O}ptimization \textbf{O}n \textbf{S}pherical \textbf{T}opology (abbreviated as \textbf{BOOST}).

Let $\mathcal{G} = (V^1,F^1)$ be the standard mesh where SBN $\mathcal{F}_{\theta}$ works, and $\mathcal{F}_{\theta}(\mu, p_1, p_2)$ be the structured output of SBN given $\mu=\{ \mu_v \colon v\in V^1\}$ and two fixed points $p_1, p_2$. 

Given a spherical surface mesh $S$, we divide it into two hemispheres, upper $S_{upp}$ and lower ones $S_{low}$ based on the signs of $z$-coordinates. Apply south ($P_S$) and north ($P_N$) pole stereographic projection to the $S_{upp}$ and $S_{low}$, respectively, we obtain  2D triangular mesh $D_{upp} \in \mathbb{R}^{|V_{upp}| \times 2}$ and $D_{low} \in \mathbb{R}^{|V_{low}| \times 2}$ within the unit disk $\mathbb{D}$. SBN, operating on a standard mesh $\mathcal{G} = (V^1,F^1)$, predicts mapping $f^S=\mathcal{F}_{\theta}(\mu_{S},p_{S,1}, p_{S,2})$ and $f^N=\mathcal{F}_{\theta}(\mu_{N},p_{N,1}, p_{N,2})$. With barycentric interpolation, we can obtain deformation of $D_{upp}$ and $D_{low}$, and finally use inverse projection $P_S^{-1}, P_{N}^{-1}$ to reconstruct the hemispheres. Sometimes, for simplicity, we may use spherical barycentric interpolation as an alternative. Essentially, the key point is to find a smooth and bijective sphere self-mapping that works on the standard spherical mesh $P_{S}^{-1}(V^1) \cup P_{N}^{-1}(V^1 \backslash \partial V^1)$ such that the interpolated result minimizes the loss objective.

The optimization needs to ensure: (1) no overlaps when gluing the hemispheres, and (2) minimal boundary mismatch. Specifically, there must be no overlaps around the boundaries of two transformed hemispheres when gluing back, meanwhile the distance between $P_S^{-1}(f^S(v))$ and $P_N^{-1}(f^N(v))$ for each $v\in \partial V^1$ should be small enough, ideally zero. For the second goal, we propose a boundary matching loss
\[
\mathcal{L}_{\text{bm}}(f^S, f^N) = \frac{1}{|\partial V^1|} \sum\limits_{v \in V^1}\| P_S^{-1}(f^S(v)) - P_N^{-1}(f^N(v)) \|_2^{2}
\]
To achieve our first aim of guaranteeing non-overlapping spherical parameterizations when gluing two deformed hemispheres, we employ a folding detection and penalization strategy based on consistent face orientation. Given the initial non-overlapping spherical parameterization $\mathbf{x}$, we first establish a consistent orientation for all triangular faces in the mesh. For each face $T$ with vertices indexed as $\{i,j,k\}$, we compute the unnormalized face normal vector $\mathbf{n}_T$ via the cross product of edge vector: $\mathbf{n}_T = ({\mathbf{x}}_{i} - {\mathbf{x}}_{k}) \times ({\mathbf{x}}_{j} - {\mathbf{x}}_{i})$. The orientation is correct if $\mathbf{n}_T \cdot \mathbf{x}_i > 0$, indicating that the normal vector points outward from the sphere's center. For faces where this condition is violated, we swap the vertex indices $j$ and $k$ to reverse the winding order, thereby ensuring all faces maintain a consistent outward-pointing orientation. With consistent orientation established, any subsequent spherical mapping can be validated for self-intersections or folding by computing the signed area of each triangular face. The signed area $A_T^s$ for an oriented face $T=(i,j,k)$ is computed as: $A_T^s = \frac{1}{2} \|\mathbf{n}_T\|_2 \cdot \text{sign}(\mathbf{n}_T \cdot \mathbf{x}_i)$.
Under the corrected orientation, faces maintaining positive signed areas ($A_T^s > 0$) iff they preserve the non-overlapping property. To explicitly discourage folding during optimization, we incorporate a folding penalty term into our loss function:
$$\mathcal{L}_{\text{folding}} = \frac{1}{|F^1|} \sum_{T \in F^1} \max(0, -A_T^s)$$
The folding penalty $\mathcal{L}_{\text{folding}}$ is assigned a higher weight than other regularization terms, because bijectivity is the fundamental topological constraint rather than a soft geometric preference. Similar folding penalty terms have been employed in various spherical mapping works as auxiliary losses to discourage self-intersections such as \cite{ren2024sugar}. Notably, our two-hemisphere gluing strategy has a significant advantage in achieving $\mathcal{L}_{\text{folding}} = 0$ compared to existing approaches. Since each hemisphere is predicted by SBN which always inherently produces homeomorphic mappings, face flipping often occurs exclusively at the deformed equator where the two charts are overlapped. The majority of faces maintain correct orientation by construction. In contrast, others must regularize all faces across the entire sphere, as inversions can occur anywhere during optimization. Therefore, it is easier to remain the bijectivity of the result mapping in our works than others and empirically final mappings by BOOST have no foldings.


Next, replace $P_N^{-1}(f^N(\partial V^1))$ in $P_N^{-1}(f^N(V^1))$ by $P_S^{-1}(f^S(\partial V^1))$, and denote this new image by $\widetilde{\mathcal{M}}$. The last thing is to ensure the smoothness of the boundary of the structured mesh $\widetilde{\mathcal{M}}$, because its boundary points and inner points are from distinct outputs.
Let $L_1$ be the circle of points inside $f^S(V^1)$ which is adjacent to the boundary points $f^S(\partial V^1)$ and $L_2$ be the circle of points inside $f^N(V^1)$ which is adjacent to the boundary points $f^N(\partial V^1)$, we define a new mesh $\mathcal{N}$ by combining $P_N(\widetilde{\mathcal{M}})$ with $P_N \circ P_S^{-1}(L_1)$. By the construction of $\mathcal{N}$, we know $P_N \circ P_S^{-1}(f^{S}(\partial V^1))$ and $P_N \circ P_S^{-1}(L_1)$ are results from the upper hemisphere, $L_2$ are results from the lower hemisphere, so for the ring of a point $v \in P_N \circ P_S^{-1}(f^{S}(\partial V^1))$, all of its adjacent faces are not output from the same hemisphere. Since we deform each hemisphere independently, the mapped positions of the one-ring neighborhood of these points in the parameter space may not align consistently. Hence we impose a boundary smoothness loss on the Laplacian and double Laplacian of the boundary points.
\[
\mathcal{L}_{\text{bs}} = \sum\limits_{v \in P_N \circ P_S^{-1}(f^{S}(\partial V^1))}  \|\Delta^2 v \|_{2}^2+ 0.1 \cdot  \|\Delta v \|_{2}^2
\]
where $\Delta v_i = \sum_{[i,j]}\omega_{ij}(v_j-v_i)$ and $\omega_{ij} = \frac{1}{2}(\cot \alpha_{ij}+\cot\beta_{ij})$, $\alpha_{ij}$ and $\beta_{ij}$ are the angles opposite to the edge $[i,j]$, and $\Delta^2$ is the standard scalar Laplacian applied twice. The double Laplacian term distributes the derivative change smoothly into the interior, preventing tension discontinuities. An illustration of $\mathcal{L}_{\text{bm}},\mathcal{L}_{\text{folding}}$ and $\mathcal{L}_{\text{bs}}$ is given in Figure \ref{fig:illustration of auxiliary loss}.
\begin{figure}
    \centering
    \includegraphics[width=\linewidth]{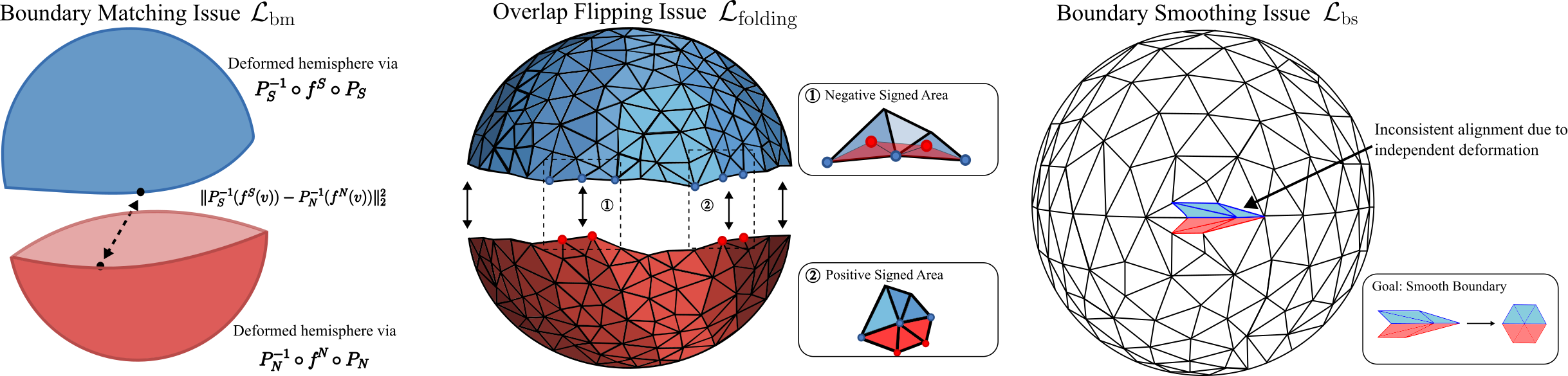}
    \caption{\small \textit{Illustration of three issues addressed by the loss functions $\mathcal{L}_{\text{bm}},\mathcal{L}_{\text{folding}}$ and $\mathcal{L}_{\text{bs}}$}}
    \label{fig:illustration of auxiliary loss}
\end{figure}

In all, let us denote the post-composition of similarity transformation by $g(x,\phi,\varphi,r) = \varphi e^{i\phi}x+r$ and after training of SBN is done, we freeze its parameters and obtain $\mathcal{F}_{\theta^*}$, whose output $\mathcal{F}_{\theta^*}(\{\mu_v\},p_1,p_2)\approx \mathcal{F}(\{\mu_{F}\},p_1,p_2,p_1,p_2)$. And we can reformulate the framework \ref{eq:new formulation of the optimization problem, formulation 3} and obtain BOOST as follows
\begin{equation}
\begin{aligned}
\min_{(T_{\text{BC}}, T_{\text{pin}},\theta_N,\theta_S)}\quad & \lambda_{\text{task}}\mathcal{L}_{\text{task}}(f^{S},f^{N}) + \lambda_{\text{reg}}\mathcal{L}_{\text{reg}}(\theta_S,\theta_N) \\&+ \big(\lambda_{\text{bm}}\mathcal{L}_{\text{bm}} + \lambda_{\text{bs}}\mathcal{L}_{\text{bs}} + \lambda_{\text{folding}}\mathcal{L}_{\text{folding}}\big)(f^S,f^N)\\
\end{aligned}
\end{equation}
where $\theta_N = \{\tilde{\mu}_N,T_{\text{BC},N},\tilde{p}_{N,1},\tilde{p}_{N,2}, T_{\text{pin},N}, \phi_N,s_N,r_N \}$ and $\theta_S = \{\tilde{\mu}_S,T_{\text{BC},S},\tilde{p}_{S,1},\tilde{p}_{S,2},T_{\text{pin},S},$ $\phi_S,\varphi_S,r_S \}$ are the parameter sets and 
\begin{align*}
    f^S & = g \big( \mathcal{F}_{\theta^*}(\mu_S,p_{S,1},p_{S,2}),\phi_S,\varphi_S,r_S \big), \\
    f^N &= g \big( \mathcal{F}_{\theta^*}(\mu_N,p_{N,1},p_{N,2}),\phi_N,\varphi_N,r_N \big).
\end{align*}
Last, the regularization term $\mathcal{L}_{\text{reg}}$ is designed to discourage local geometric distortion and enhance the smoothness of the mapping. Therefore, two regularization terms are always used in $\mathcal{L}_{\text{reg}}$, defined as 
\begin{align*}
    \mathcal{L}_{\text{BC}} &= \frac{1}{|V^1|}\sum\limits_{v\in V^1}\left( |\mu_{S}(v)|^2 +|\mu_{N}(v)|^2\right) \\
    \mathcal{L}_{\text{smooth}} = & \frac{1}{|F^1|}\sum\limits_{T\in F^1}\left( |\nabla \mu_{S}(T)|^2 + |\nabla \mu_{N}(T)|^2 \right) 
\end{align*}
The details discussed are summarized in Algorithm \ref{algo: joint lsqc}.
\begin{algorithm}[htb]
    \caption{Beltrami Optimization On Spherical Topology}
    \label{algo: joint lsqc}
    \begin{algorithmic}[1]
        \Require
            Data: A triangular mesh of unit sphere $\mathbb{S}^2$.
            \begin{enumerate}
            \item optimization parameters $\eta$: $\theta_N = \{\tilde{\mu}_N,T_{\text{BC},N},\tilde{p}_{N,1},\tilde{p}_{N,2}, T_{\text{pin},N}, \phi_N,s_N,r_N \}$ and $\theta_S = \{\tilde{\mu}_S,T_{\text{BC},S},\tilde{p}_{S,1},\tilde{p}_{S,2},T_{\text{pin},S},$ $\phi_S,\varphi_S,r_S \}$, 
            \item loss weights $\lambda_j$, loss functions $\mathcal{L}_{j}$, optimizer $\Psi$
            \end{enumerate}
        
        \Ensure
            transformed spherical mesh $\tilde{S}$
        \State 
            Freeze the parameters of SBN $\mathcal{F}_{\theta^*}$.
       
        \State 
            $continue$ = $True$;
        \While{ $continue$ }
            \State
            Prepare input $\mu = \mathcal{T}(\tilde{\mu},T_{\text{BC}}), p_i = \mathcal{T}(\tilde{p}_i,T_{\text{pin}})$ for the model $\mathcal{F}$.
            \State
            Compute two mappings predicted by model, $f^S =g \big( \mathcal{F}_{\theta^*}(\mu_S,p_{S,1},p_{S,2}),\phi_S,\varphi_S,r_S \big)$ and 
            $f^N = g \big( \mathcal{F}_{\theta^*}(\mu_N,p_{N,1},p_{N,2}),\phi_N,\varphi_N,r_N \big)$.
            \State Obtain the deformed spherical surface mesh $\tilde{S}$ via interpolation on $f^S,f^N$.
            \State
            Compute the weighted sum $\mathcal{L}_{\text{total}}$ of losses $\mathcal{L}_{\text{task}}(\tilde{S})$, $\mathcal{L}_{\text{reg}}(\theta_S,\theta_N)$, $\mathcal{L}_{\text{bm}}$, $\mathcal{L}_{\text{bs}}$, $ \mathcal{L}_{\text{folding}}$ and $\mathcal{L}_{\text{smooth}}$ with weights $\lambda_j$.
    
            \If{$\mathcal{L}_{\text{total}}$ and each individual loss meet stop criterion}  
                \State
                    $continue$ = $False$; 
            \Else
                \State 
                Compute $\nabla_{\eta}\mathcal{L}_{total}$ and update all $\eta$ with the optimizer $\Psi$.
            \EndIf
        \EndWhile
    \end{algorithmic}
\end{algorithm}

\section{Experiments and Applications}
\label{sec: experiments}
Numerical experiments were conducted to evaluate the efficacy of our framework BOOST, which comprises several applications in landmark, intensity, and hybrid registration problems, encompassing both realistic and synthetic scenarios. For the training detail of the SBN model, please refer to \cite{xu2025neural}.
In downstream applications, all experiments related to BOOST were conducted on an NVIDIA A40 GPU with 44GB memory, ensuring efficient computation, and the numerical algorithms were implemented using MATLAB on an 8-core Intel i7 machine.
\subsection{Experiments}
We first tested our methods on two synthetic experiments.
In the first test, the method was applied to a hybrid registration task requiring a bijective deformation that morphs an 'I'-shaped region on the sphere into a 'C' shape, while matching six pairs of corresponding landmarks and the entire image intensity. As shown in Figure \ref{fig:deform_I_to_C}, the deformed image very closely resembles the target image, landmarks are consistently matched and the obtained registration remains bijective.

We further challenged our method by conducting a landmark-matching experiment involving extreme distortions. Specifically, we randomly selected $2n \ (n = 2,3,4)$  groups of landmarks $(A_{i},B_i,C_i,D_i)_{i=1,2,...,2n}$, distributed evenly on both hemispheres. Each group of landmarks is constrained that adjacent points were mapped to each other, thereby constructing $2n$ large-scale twists on the spherical surface. The top-left picture of Figure \ref{fig:extreme distortion} displays representative pairs of landmarks, while the task loss in this setup is defined as
\begin{equation*}
    \mathcal{L}_{task}(f) =\frac{1}{8n}\sum\limits_{i=1}^{n}(|f(A_i)-B_i|^2 +|f(B_i)-C_i|^2+|f(C_i)-D_i|^2+|f(D_i)-A_i|^2) .
\end{equation*}
The Top-right picture of Figure \ref{fig:extreme distortion} shows the final registration grid for a case that $n=4$ and the average L2 squared error of each pair is decreased to $7.96 \times 10^{-4}$  after optimization, and the second row is the Beltrami distribution of the result mapping. As demonstrated by the qualitative and quantitative results, SBN is highly adept at dealing with tasks requiring extremely large deformation mappings: bijectivity is rigorously preserved since all Beltrami coefficients retain a norm less than 1 throughout, while smoothness and conformality are maintained by constraining both $|\mu|$ and $|\nabla \mu|$. Thus, even under extreme landmark-driven warping, regions outside the immediate deformation zones remain minimally distorted.
\begin{figure}
    \centering
    \includegraphics[width=\linewidth]{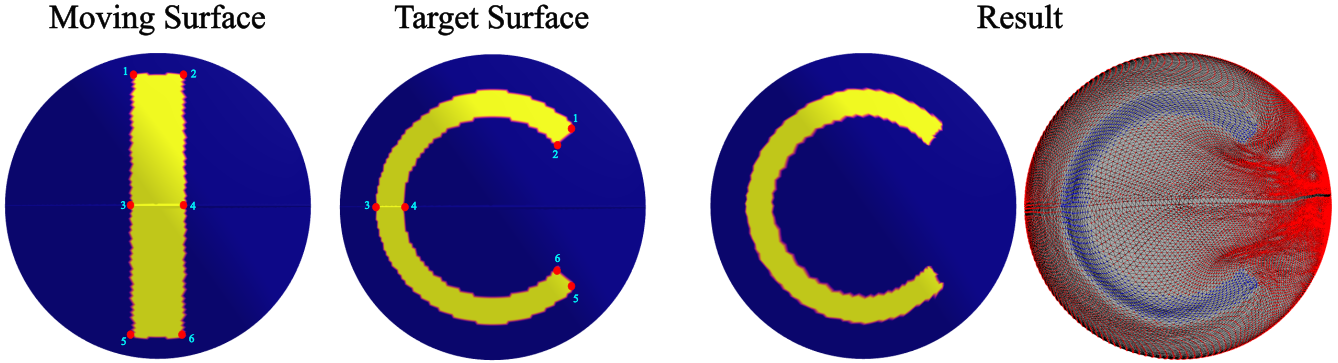}
    \caption{\small \textit{Deformation from an 'I'-shape region to a 'C'-shape region. The moving surface (showing a vertical yellow strip) is transformed to align with the target surface (displaying a C-shaped yellow region). The third panel shows the successfully registered result after applying the transformation. The rightmost visualization presents the registration grid.} }
    \label{fig:deform_I_to_C}
\end{figure}
\begin{figure}
    \centering
    \includegraphics[width=0.8\linewidth]{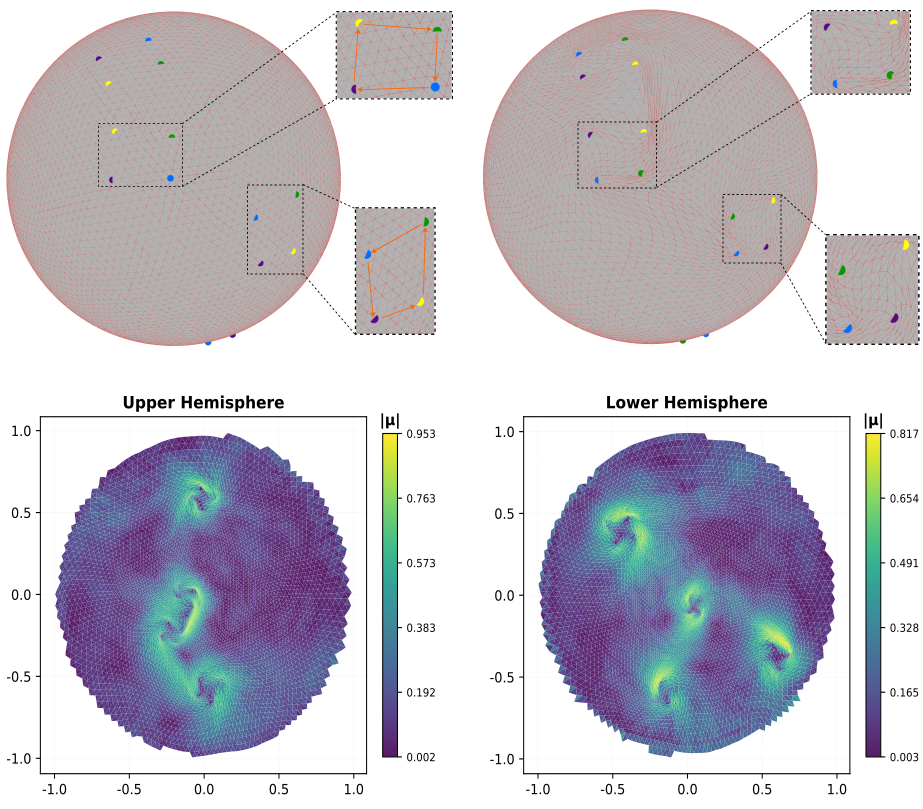}
    \caption{\small \textit{Results for the extreme distortion experiment. Top row: Landmarks to be matched, forming large-scale twists. Middle row: Deformation result of the surface after matching. Bottom row: Mapping results shown for both hemispheres, with color indicating the magnitude of the Beltrami coefficient.}}
    \label{fig:extreme distortion}
\end{figure}

\subsection{Applications to brain surface registration}
\label{subsec:cortical surface registration}
Cortical surface registration plays a crucial role in aligning functional and anatomical features across individuals. Due to its requirement on the high registration accuracy, preserving topology, and minimizing registration distortion, the problem is challenging.
To demonstrate the effectiveness of our framework, we conducted a series of experiments focusing on the registration between spherical parameterizations of cortical surfaces. 

We tested Algorithm \ref{algo: joint lsqc} by performing a landmark alignment task on the spherical parameterization of right cortical surface data. We followed the experiment setup from FLASH \cite{choi2015flash}. On each cortical surface, there are six landmark curves, including Central Sulcus (CS), Inferior Frontal Sulcus (IFS), Superior Frontal Sulcus (SFS), Inferior Temporal Sulcus (ITS), Superior Temporal Sulcus (STS) and Postcentral Sulcus (PostCS). We first used FLASH \cite{choi2015flash} to obtain a spherical conformal parameterization of each cortical surface and aligned landmarks on the pair of spherical cortical surfaces. 
Mathematically, denote six landmark curves on the spherical parameterization of two cortical surfaces by $S_{1,i}, S_{2,i} (i=1,\ldots,6)$ and similarity measure functions of landmark pairs to be matched by $d_{sim}$. Our goal is to find $f \colon \mathbb{S}^{2} \rightarrow \mathbb{S}^2$ such that 

\begin{enumerate}
    \item minimizing $d_{sim}(f(S_{1,i}),S_{2,i})$ for all $i=1.\ldots, 6$;
    \item $f$ is bijective;
    \item $f$ is geometrically regularized.
\end{enumerate}

The goal is formulated as finding an optimal mapping $f$ minimizing an energy functional
\[
\mathcal{L}_{\text{lm}}(f) = \mathcal{L}_{\text{task}}(f) + \mathcal{L}_{\text{reg}}(f),
\]
where 
\[
\mathcal{L}_{\text{task}}(f) = \frac{1}{6}\sum\limits_{i=1}^{6}d_{\text{sim}}(f(S_{1,i}), S_{2,i}).
\]

For the regularization terms involved in included but not limited to this experiment, we usually use a weighted sum of $\mathcal{L}_{\text{bm}}, \mathcal{L}_{\text{folding}}, \mathcal{L}_{\text{bs}}, \mathcal{L}_{\text{BC}}$ and $\mathcal{L}_{\text{smooth}}$.

In prior work \cite{choi2015flash}, $d_{sim}$ is nodewise L2 square distance and this is indispensable for them to linearize the optimization problem via the variational method. To satisfy this, Choi et al. \cite{choi2015flash} used resampling to ensure that the number of nodes in the corresponding curves is the same. We applied the same resampling operation, followed by a rotation to minimize the average landmark L2 distance. The registration was then performed using our SBN-based framework, and the results were compared with those of FLASH\cite{choi2015flash} in terms of both angle distortion and landmark matching error. The hyperparameter $\lambda$ in FLASH was set to 3 and the weights of $\mathcal{L}_{\text{task}}, \mathcal{L}_{\text{bm}}, \mathcal{L}_{\text{folding}}, \mathcal{L}_{\text{bs}}, \mathcal{L}_{\text{BC}}$ and $\mathcal{L}_{\text{smooth}}$ were 5, 1, 20, 0.5, 0.1 and 0.01, respectively.
Numerical results show that our method achieves competitive, or even superior, landmark matching accuracy meanwhile always yielding consistently lower angle distortions. Figure \ref{fig:small_deformation_move32_temp35} and Figure \ref{fig:bc_comparison_move32_temp35} compare the performance of different methods in an example of registering brain 32 to brain 35, in terms of visual effect and statistics of induced Beltrami coefficients. See Table \ref{tab:comparison in small deformation case} for a more comprehensive comparison. 
\begin{figure}[!tbp]
    \centering
    \includegraphics[width=\linewidth]{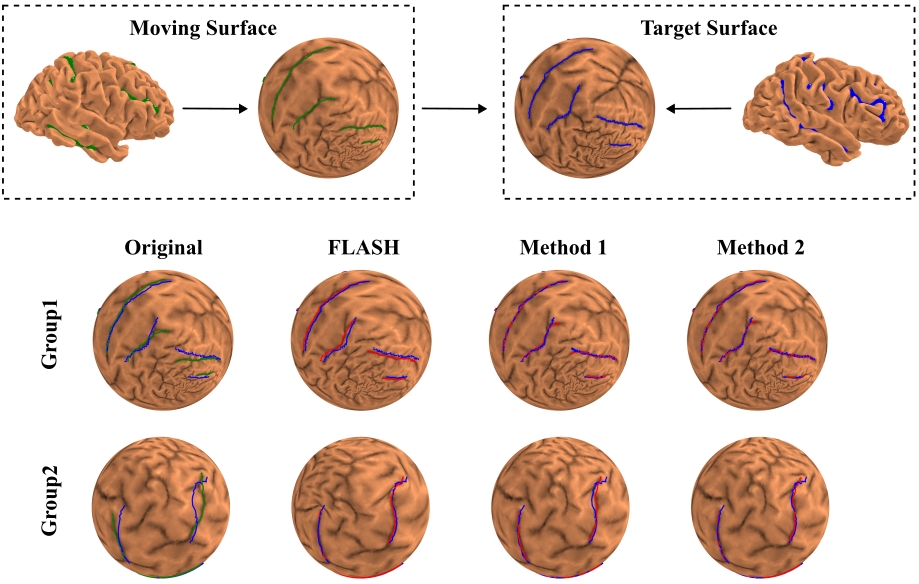}
    \caption{\small \textit{Registration results for brain 32 to brain 35 using six sulcal landmarks. Landmarks are grouped (Group 1: CS, ITS, STS, postCS; Group 2: IFS, SFS) for visual clarity. Columns, from left to right: source/target landmark pairs (green: moving surface, blue: target surface), results from FLASH\cite{choi2015flash}, results of our method with resampling, and results of our method with chamfer distance. Red curves are the deformed sulci, and blue ones shows target sulci.}}
    \label{fig:small_deformation_move32_temp35}
\end{figure}
\begin{figure}[!htbp]
    \centering
    \includegraphics[width=\linewidth]{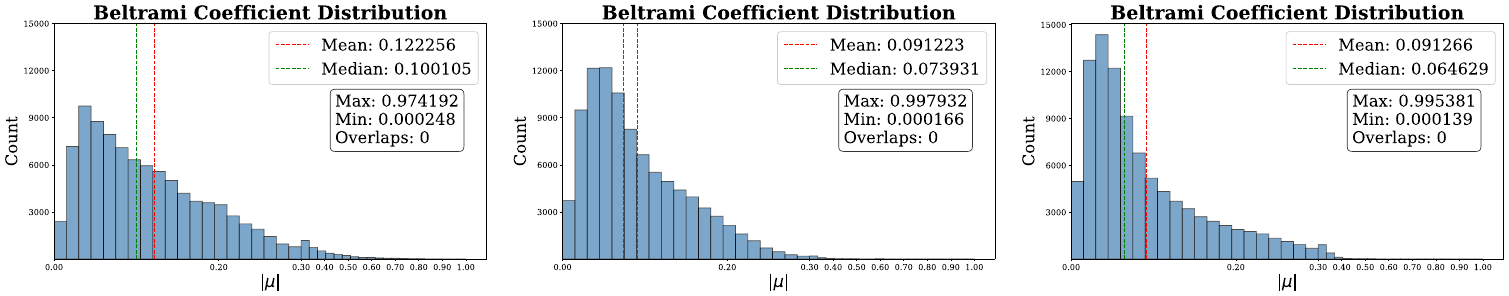}
    \caption{\small \textit{Comparison of induced angle distortion for three methods; FLASH\cite{choi2015flash} (Left), our method with landmark resampling (Middle), and our method using chamfer distance (Right), represented by distributions of Beltrami coefficients associated with deformation.}}
    \label{fig:bc_comparison_move32_temp35}
\end{figure}
\begin{table}[!htbp]
  \centering
  \resizebox{\linewidth}{!}{
  \begin{tabular}{|c|c|c|c|c|c|c|c|c|c|c|c|c|}
  \hline
  \multirow{4}{*}{\textbf{surface}} & \multicolumn{4}{c|}{\textbf{FLASH}\cite{choi2015flash}} & \multicolumn{4}{c|}{\textbf{method 1}} & \multicolumn{4}{c|}{\textbf{method 2}} \\
  \cline{2-13}
   & Number  &  Mean  & Mean  & Mean  & Number  &  Mean  & Mean  & Mean & Number  &  Mean  & Mean  & Mean \\
   & of over- & \multirow{2}{*}{$|\mu|$} & Square  & chamfer & of over- & \multirow{2}{*}{$|\mu|$} & Square  & chamfer& of over- & \multirow{2}{*}{$|\mu|$} & Square  & chamfer\\
   &laps & & distance & distance &laps & & distance & distance &laps & & distance & distance \\
    \hline
    Brains 35 and 32 & \textbf{0} & 0.123 & 0.0018  & 0.0014  & \textbf{0} & \textbf{0.09122} & \textbf{0.0007} & 0.0009  & \textbf{0} & 0.09126  & 0.0027 & \textbf{0.0007} \\
    \hline
    Brains 29 and 13 & 44 & 0.296 & 0.0034 & 0.0030  & \textbf{0} & \textbf{0.091} & \textbf{0.0018} & 0.0012  & \textbf{0} & 0.098  & 0.0056 & \textbf{0.0004} \\
    \hline
    Brains 18 and 8 & 107 & 0.307 & 0.0022 & 0.0029 &  \textbf{0} & 0.096 & \textbf{0.0011} & 0.0012  & \textbf{0} & \textbf{0.074}  & 0.0117 & \textbf{0.0007} \\
    \hline
    Brains 21 and 39  & 128 & 0.366 & 0.0411 & 0.0632  & \textbf{0} & 0.158 & \textbf{0.0018} & 0.0020  & \textbf{0} & \textbf{0.144}  & 0.0061 & \textbf{0.0009} \\
    \hline
    Brains 10 and 9  & 157 & 0.387 & 0.0048 & 0.0037  & \textbf{0} & 0.118 & \textbf{0.0015} & 0.0014  & \textbf{0} & \textbf{0.113}  & 0.0056 & \textbf{0.0003} \\
    \hline
  \end{tabular}
  }
  \caption{\small \textit{Comparison of landmark matching performance and angle distortion among FLASH\cite{choi2015flash} and BOOST (with two variations: method 1—with landmark resampling; method 2—using chamfer distance as loss, without landmark resampling). For each pair of brains, the table provides the number of overlap failures, average $|\mu|$, mean squared distance, and mean chamfer distance. Lower values are better in all columns.}}
  \label{tab:comparison in small deformation case}
\end{table}

Moreover, in both experimental and real-world brain landmark matching tasks, it is impractical to require that corresponding curves from different brains have the same number of points. Curves like the Central Sulcus of different brains may be represented by different numbers of sampled points due to variations in brain size, shape, or imaging resolution. This discrepancy poses a significant challenge for traditional alignment methods that rely on point-to-point comparisons. Additionally, if the same curves in two subjects indeed have very different lengths, as the following large deformation case indicates, these resampling methods may not accurately capture true anatomical correspondence and the traditional methods may fail.

An advantage of our framework is its flexibility to incorporate with any differentiable loss functions. By treating each pair of curves to be matched as two point clouds, we can adopt the chamfer distance as the loss function. Let $S_1$ and $S_2$ be two point clouds, and their chamfer distance is defined by
 \begin{equation*}
      d_{CD}(S_1,S_2) = \frac{1}{|S_1|}\sum\limits_{y \in S_1} \min\limits_{x \in S_2} \|x-y\|^2 + \frac{1}{|S_2|}\sum\limits_{x \in S_2} \min\limits_{y \in S_1} \|y-x\|^2.
 \end{equation*}
For better performance, we also required that two endpoints of each curve on the moving surface can be matched with those of its target. Suppose $q_{1,i}, q_{2,i}$ are endpoints of $i$-th curves in $S_1$ and $t_{1,i}, t_{2,i}$ are their targets, then 
\[
\mathcal{L}_{task}(f) = \frac{1}{6}\sum\limits_{i=1}^{6}d_{CD}(f(S_{1,i}), S_{2,i}) + \frac{1}{12}\sum\limits_{i=1}^{6}\sum\limits_{j=1}^{2}|q_{j,i}-t_{j,i}|^{2}.
\]
The weights of $\mathcal{L}_{\text{task}}, \mathcal{L}_{\text{bm}}, \mathcal{L}_{\text{folding}}, \mathcal{L}_{\text{bs}}, \mathcal{L}_{\text{BC}}$ and $\mathcal{L}_{\text{smooth}}$ in this method were the same as before.
For clarity, we denote the former application of our framework by method 1, and this by method 2. A result is shown in the rightmost column of Figure \ref{fig:small_deformation_move32_temp35}. As shown in Table \ref{tab:comparison in small deformation case}, replacing the L2 nodewise distance by the sum of the chamfer distance and endpoint loss consistently achieves visually better landmark matching performance and smaller angle distortion than FLASH \cite{choi2015flash}.

In addition, even when the initial landmark distortion is small, FLASH \cite{choi2015flash} may fail to achieve a bijective deformation when imposing too many landmark constraints, whereas our method remains robust. For example, when registering brain 8 (moving surface) to brain 18 (template surface), FLASH is able to handle registering ITS and STS, or CS, SFS, IFS, and postCS separately, but cannot provide a non-overlap registration when all six sulci are involved. In contrast, our method can simultaneously register all six landmark curves, and the registration accuracy for all six sulci, exceeds that of FLASH even when FLASH handles these sulci in individual smaller groups. But for a fair comparison of the capacity of controlling geometric distortion, in Figure \ref{fig:move8_temp18} we compare the registration result when all methods are required to register only CS, SFS, IFS, and postCS. Compared to FLASH, we can match the landmark curve well without inducing too much distortion in other regions. As shown in Figure \ref{fig:move8_temp18}, the registration result by FLASH is highly biologically implausible, while ours induce much less distortion in regions away from landmark curves, no matter whether using landmark resampling or not.
\begin{figure}[htb]
    \centering
    \includegraphics[width=\linewidth]{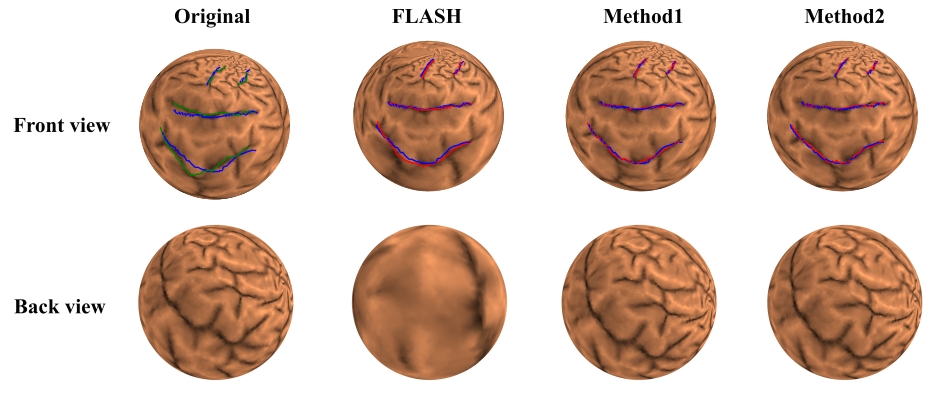}
    \caption{\small \textit{Registration result comparison of brain 8 to brain 18 with 4 sulcal landmarks in two views (for visualization and color conventions, see Figure  \ref{fig:small_deformation_move32_temp35}).}}
    \label{fig:move8_temp18}
\end{figure}
Another advantage of BOOST is that it remains robust in the large deformation case. See Figure \ref{fig:move39_temp35} for an example registering CS, IFS and SFS of brain 39 to those of brain 35. As shown, there are big differences between lengths and positions between two curves in each landmark pair. FLASH \cite{choi2015flash} fails to give a non-overlap deformation while ours still provides a bijective registration result with little distortion. To demonstrate robustness with complex landmark configurations, we performed registration on 11 sulci pairs annotated on two left hemisphere cortical surfaces, including the lateral sulcus, central sulcus, postcentral sulcus, intraparietal sulcus, transverse occipital sulcus, inferior frontal sulcus, precentral sulcus, superior temporal sulcus, inferior temporal sulcus, superior frontal sulcus, and lateral occipital sulcus. While FLASH failed to maintain bijectivity under this dense landmark constraint, both proposed methods successfully produced non-overlapping mappings with accurate landmark correspondence (see Figure \ref{fig: OASIS_11_sulci_landmark_registration_comparison}).
\begin{figure}[!tbp]
    \centering
    \includegraphics[width=\linewidth]{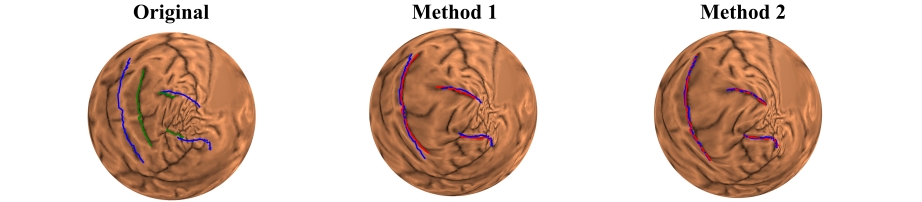}
    \caption{\small \textit{Registration of brain 39 to brain 35 using CS, IFS, SFS as landmark curves (for color conventions, see Figure \ref{fig:small_deformation_move32_temp35}). The first column is visualization on the original moving face. The second and third ones are results using method 1 and 2, respectively.}}
    \label{fig:move39_temp35}
\end{figure}
\begin{figure}
    \centering
    \includegraphics[width=\linewidth]{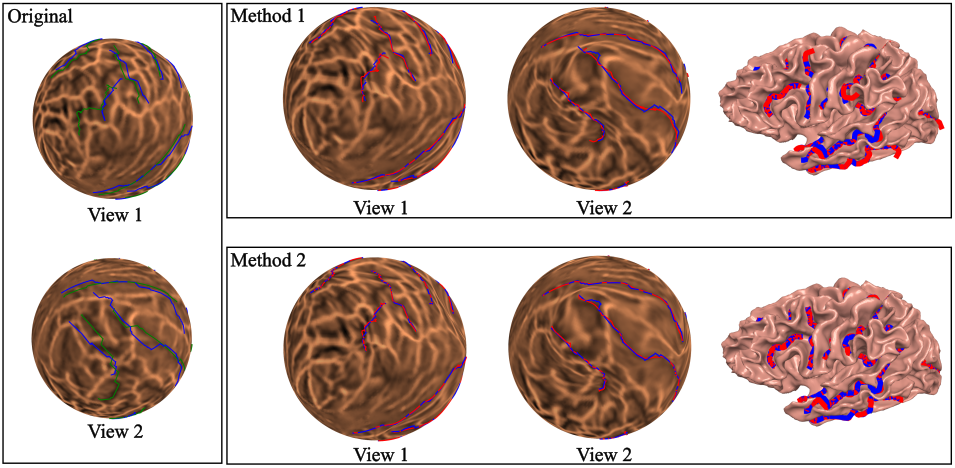}
    \caption{\small \textit{Visualization of registration results by our two methods. Left: Original landmark configuration before registration shown on spherical parameterizations (View 1 and View 2), where blue indicates target landmarks and green indicates moving surface landmarks. Method 1 (top row) and Method 2 (bottom row): Registration results displayed on spherical representations and mapped to the actual cortical surface. In both methods' results, red landmarks represent the deformed moving surface landmarks after registration, while blue landmarks denote the target positions.}}
    \label{fig: OASIS_11_sulci_landmark_registration_comparison}
\end{figure}

Furthermore, our method is seamlessly compatible with  alignment of geometric feature and anatomical parcellation. For illustration, we registered 189 subjects' right-hemi cortical surface in the OASIS3 dataset \cite{lamontagne2019oasis} to population averages from FreeSurfer\cite{fischl2012freesurfer}. More specifically, initial spherical representations for each subject's cortical surface were generated using FreeSurfer. We used the fsaverage6 surface, a standard population-average spherical template provided by FreeSurfer, as the common registration target space. For the parcellation registration, we used FreeSurfer to generate parcellation labels (i.e. 35 cortical areas) for each subject including the fsaverage6 surface. 
Let $\Phi$ be the non-rigid deformation, and $M,F$ be geometric features on the moving sphere and the fixed sphere, respectively. During optimization, the similarity of features between  $M \circ \Phi$ and $F$ is estimated using the Normalized correlation coefficient(NCC), 
\[
d_{corr}(M \circ \Phi, F) = \frac{cov(M \circ \Phi, F)}{\sigma(M \circ \Phi) \sigma(F)} .
\]
For the parcellation alignment, we leveraged the Dice score of anatomical parcellation between the moved and fixed spheres:
\[
\text{Dice}(M^p\circ\Phi,F^p) = \frac{2 \cdot |M^p\circ\Phi \cap F^p|}{|M^p\circ\Phi|+|F^p|},
\]
where $M^p$ and $F^p$ are the parcel $p$ in the moving and fixed spheres. The parcellation loss is 
\[
d_{par}(M \circ \Phi, F) = \frac{1}{P}\sum_{i=1}^{P}(1-\text{Dice}(M^p\circ\Phi,F^p)),
\]
where $P$ denotes the number of parcels, and the parcellation atlas used in registration was generated by FreeSurfer. 

We applied BOOST with $\mathcal{L}_{task} = 10 \cdot d_{par}+ d_{corr}$ for the experiment. The parcellation accuracy was quantified using the Dice score (higher values denote greater overlap). The alignment of geometric features was evaluated by comparing each subject’s sulcal depth to the FreeSurfer fsaverage atlas. Performance was measured by two metrics, NCC and Dice score, with higher NCC values indicating better alignment accuracy in sulcal depth and higher Dice score meaning greater parcellation similarity. The mean and standard deviation of NCC and Dice score are 0.9122 ± 0.0169 and 0.9216 ± 0.0262, respectively. Figure \ref{fig:sulc_parc_reg} shows a more detailed result for 35 regions parceled by the atlas provided by FreeSurfer.  
It is worth mentioning that in SUGAR \cite{ren2024sugar}, Ren et al. conduct the same experiment on the ADRC dataset, which also consists of  brain data from Alzheimer Disease subjects, and their method achieved 0.904 NCC and 0.839 Dice score. This implies that our method is also competitive and potentially superior to other mainstream algorithms.

Lastly, we applied our algorithm to registering sulcal depth and six sulci of two cortical surfaces. We used the left cortical surface of subject 30001 and 30003 in the OASIS3 dataset as the moving surface and target surface, respectively, and matched the six pairs of landmark curves meanwhile maximizing the correlation of their sulcal depth. Figure \ref{fig:sulcal_depth_correlation_comparison} is the distribution of sulcal depth on cortical surfaces, which shows that distribution of sulcal depth on the moving surface is more similar to that on the fixed surface after deformation. Notably, after hybrid matching, the similarity in sulcal depth distribution is improved from a pre-registration correlation coefficient of 0.490 to 0.705 post-registration (as opposed to just 0.166 when performing only landmark matching). We also compare the landmark matching effect of these two methods in Figure \ref{fig: hybrid_register_OASIS_brain_surf_compariso}, and hybrid matching delivers both appropriate landmark registration and sulcal depth alignment, enhancing the accuracy and plausibility of surface correspondence, which is an advantage in quantitative neuroimaging analysis and medical diagnosis. 

\begin{figure}
    \centering
    \includegraphics[width=\linewidth]{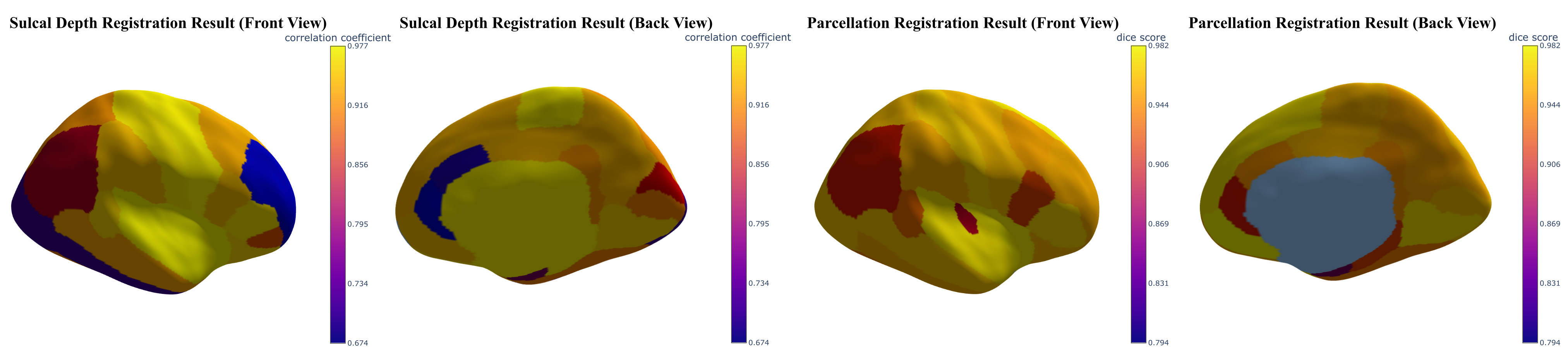}
    \caption{\small \textit{Quantitative Evaluation of Surface Registration Accuracy for Parcellation and Sulcal Depth Data. The left panels show parcellation registration results with Dice scores, indicating regional alignment accuracy across different brain parcels. The right panels present sulci depth registration results with correlation coefficients, measuring the correspondence of sulcal depth patterns between registered surfaces and the averaged template.}}
    \label{fig:sulc_parc_reg}
\end{figure}

\begin{figure}[!tbp]
    \centering
    \includegraphics[width=\linewidth]{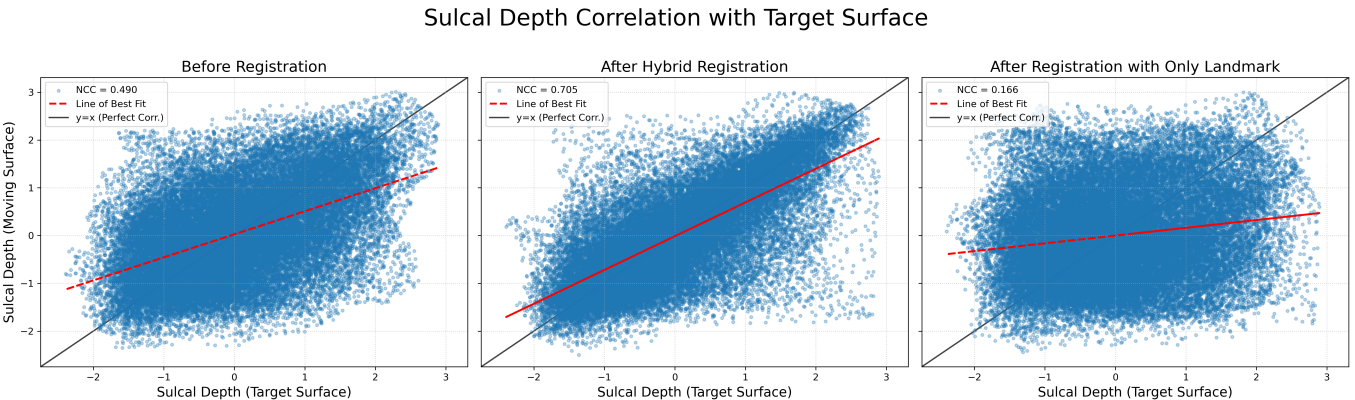}
    \caption{\small \textit{Sulcal depth correlations between the moving surface and the target surface before and after registration. From left to right: (1) before deformation; (2) after hybrid matching registration; (3) after only-landmark matching registration. Alignment is visually and statistically improved after hybrid matching.}}
    \label{fig:sulcal_depth_correlation_comparison}
\end{figure}
\begin{figure}[!tbp]
    \centering
    \includegraphics[width=\linewidth]{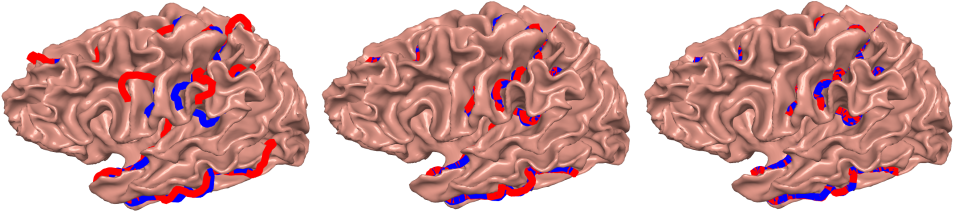}
    \caption{\small \textit{Comparing registration results with sulcal landmarks of the moving surface highlighted in red and the target landmarks highlighted in blue: (left) fixed cortical surface, (center) deformed moving surface after hybrid matching, (right) deformed moving surface using landmark-only matching. Hybrid matching achieves both precise landmark alignment and improved sulcal depth correlation.}}
    \label{fig: hybrid_register_OASIS_brain_surf_compariso}
\end{figure}

\section{Discussion}
\label{Sec: discussion}
In this paper we have introduced the concept of the Spherical Beltrami Differential, which has a one-to-one correspondence with the equivalence class of spherical homeomorphisms up to the group of conformal mappings. Following this, we used the Spherical Beltrami Network (SBN) to establish a neural optimization framework BOOST for spherical parameterization problems, which achieves optimization of spherical Beltrami differentials in a two-chart manner. Across multiple distinct testbeds including cortical surface alignment, extreme distortions, and hybrid intensity‐landmark matching, BOOST matched or beat both classical numerical algorithms. Here, we try to analyze several potential factors that underlie this excellence:
\begin{enumerate}
    \item \textbf{Power of SBN.} The effectiveness of BOOST stems directly from the capacity of the LSQC neural surrogate, the Spectral Beltrami Network (SBN). By closely approximating the LSQC objective, SBN ensures that for an admissible Beltrami coefficient (BC) field, the predicted mapping tracks the true LSQC solution and usually remains bijective, thereby realizing the deformation prescribed by the input BC. This high-quality differentiable approximation gives BOOST a smooth, well-conditioned optimization landscape in the Beltrami space, enabling gradient-based updates and suited for general spherical parameterization problems even some cases that explicit gradient on the deformation mapping is hard to compute.
    \item \textbf{Free-boundary optimization.}  Traditional numerical algorithms usually impose hard boundary constraints, which shrinks the admissible solution space and increases the risk of getting trapped in sub-optimal local minima. SBN inherits the free-boundary perspective of LSQC energy, which provides a broader search space of mappings. Moreover, the discrete choice of pinpoint required in LSQC becomes differentiable in SBN, i.e. BOOST automatically learns what is the most appropriate pinpoint condition, and thus BOOST can be applied in landmark-free mapping problems.
    \item \textbf{Direct control on conformality.} While classical methods often minimize angle distortion indirectly (e.g. Laplacian smoothing) and deep learning approaches can control angle distortion only through penalizing the output mapping, our BOOST, by regularizing the Beltrami differential directly via smoothness and magnitude penalties, achieves fine‐grained angular‐distortion control.
\end{enumerate}
The excellence of SBN lies not just in its use of deep learning, but that it faithfully encodes a mathematically well-behaved energy whose properties—existence, uniqueness, invariance, and resolution-independence, which translate directly into superior performance. By these, SBN achieves what neither classical solvers nor generic neural networks can: produce accurate surface mappings with well-controlled distortion and generalize across resolutions and different applications. Though SBN represents a substantial advance and BOOST outperforms in applications, several pragmatic gaps remain:
\begin{enumerate}
    \item \textbf{Not totally mesh-free.} Currently, to handle ultra-high-resolution meshes, we interpolate the solution from a lower-resolution “standard” grid, which is supported by the resolution-independence property of LSQC. However, this property is built on the assumption that Beltrami coefficients on the three sub-triangles equal to that of the parent triangle before splitting. This condition might be inappropriate in high-curvature or feature-dense regions, and can introduce interpolation artifacts that do not faithfully reflecting fine-mesh geometry.  An ideal remedy is a truly resolution-invariant architecture (e.g. continuous operators or adaptive patch hierarchies) that naturally receives Beltrami data and outputs deformation at arbitrary granularity, and that would be our future working direction.
    \item \textbf{Limited to genus-0 topology.} Our current framework BOOST assumes genus-0 (spherical) topology. Extending to higher-genus or open surfaces is straightforward in principle—by redefining shape of parameterization domain as well as chart overlaps, and this is also left to future work.
\end{enumerate}
\section{Conclusions}
\label{sec: conclusions}
In this work, we introduced the novel concept of the Spherical Beltrami Differential and built a correspondence between Spherical Beltrami Differentials and spherical self-quasiconformal mappings. Following this, we proposed a two-chart optimization framework BOOST tailored to complex genus-0 surface mapping problems. As a key factor in BOOST, SBN can produce free-boundary quasiconformal mappings accurately following the input BCs. By designing three auxiliary losses $\mathcal{L}_{\text{bm}},\mathcal{L}_{\text{bs}} \text{ and }\mathcal{L}_{\text{folding}}$, BOOST resolves gluing issues between two hemispheres, and achieves smooth transitions across hemispherical boundaries. These capabilities establish BOOST as a robust and versatile framework for spherical parameterization problems. The excellence of BOOST has been demonstrated through extensive experiments and applications across three key categories: Landmark matching, Intensity matching and Hybrid matching registration. 

To the best of our knowledge, this is the first work to achieve optimization over surface Beltrami differentials, and the outstanding performance of BOOST implies the potential of exploring optimization over surface Beltrami Differentials on general surfaces.

\newpage 
\bibliographystyle{siamplain}
\bibliography{references}

@article{li2024bi,
  title={A bi-variant variational model for diffeomorphic image registration with relaxed Jacobian determinant constraints},
  author={Li, Yanyan and Chen, Ke and Chen, Chong and Zhang, Jianping},
  journal={Applied Mathematical Modelling},
  volume={130},
  pages={66--93},
  year={2024},
  publisher={Elsevier}
}

@article{chen2024multiscale,
  title={Multiscale approach for variational problem joint diffeomorphic image registration and intensity correction: theory and application},
  author={Chen, Peng and Chen, Ke and Han, Huan and Zhang, Daoping},
  journal={Multiscale Modeling \& Simulation},
  volume={22},
  number={3},
  pages={1097--1135},
  year={2024},
  publisher={SIAM}
}

@article{chen2024three,
  title={Three-Stage Approach for 2D/3D Diffeomorphic Multimodality Image Registration with Textural Control},
  author={Chen, Ke and Han, Huan},
  journal={SIAM Journal on Imaging Sciences},
  volume={17},
  number={3},
  pages={1690--1728},
  year={2024},
  publisher={SIAM}
}

@article{fischl1999high,
  title={High-resolution intersubject averaging and a coordinate system for the cortical surface},
  author={Fischl, Bruce and Sereno, Martin I and Tootell, Roger BH and Dale, Anders M},
  journal={Human brain mapping},
  volume={8},
  number={4},
  pages={272--284},
  year={1999},
  publisher={Wiley Online Library}
}

@article{robinson2014msm,
  title={MSM: a new flexible framework for multimodal surface matching},
  author={Robinson, Emma C and Jbabdi, Saad and Glasser, Matthew F and Andersson, Jesper and Burgess, Gregory C and Harms, Michael P and Smith, Stephen M and Van Essen, David C and Jenkinson, Mark},
  journal={Neuroimage},
  volume={100},
  pages={414--426},
  year={2014},
  publisher={Elsevier}
}

@article{cheng2020cortical,
  title={Cortical surface registration using unsupervised learning},
  author={Cheng, Jieyu and Dalca, Adrian V and Fischl, Bruce and Z{\"o}llei, Lilla and Alzheimer’s Disease Neuroimaging Initiative and others},
  journal={NeuroImage},
  volume={221},
  pages={117161},
  year={2020},
  publisher={Elsevier}
}

@inproceedings{suliman2022deep,
  title={A deep-discrete learning framework for spherical surface registration},
  author={Suliman, Mohamed A and Williams, Logan ZJ and Fawaz, Abdulah and Robinson, Emma C},
  booktitle={International Conference on Medical Image Computing and Computer-Assisted Intervention},
  pages={119--129},
  year={2022},
  organization={Springer}
}

@article{praun2003spherical,
  title={Spherical parametrization and remeshing},
  author={Praun, Emil and Hoppe, Hugues},
  journal={ACM transactions on graphics (TOG)},
  volume={22},
  number={3},
  pages={340--349},
  year={2003},
  publisher={ACM New York, NY, USA}
}

@inproceedings{wan2012efficient,
  title={Efficient spherical parametrization using progressive optimization},
  author={Wan, Shenghua and Ye, Tengfei and Li, Maoqing and Zhang, Hongchao and Li, Xin},
  booktitle={International Conference on Computational Visual Media},
  pages={170--177},
  year={2012},
  organization={Springer}
}

@inproceedings{kazhdan2012can,
  title={Can mean-curvature flow be modified to be non-singular?},
  author={Kazhdan, Michael and Solomon, Jake and Ben-Chen, Mirela},
  booktitle={Computer Graphics Forum},
  volume={31 Issue 5},
  pages={1745--1754},
  year={2012},
  organization={Wiley Online Library}
}

@article{cui2019spherical,
  title={Spherical optimal transportation},
  author={Cui, Li and Qi, Xin and Wen, Chengfeng and Lei, Na and Li, Xinyuan and Zhang, Min and Gu, Xianfeng},
  journal={Computer-Aided Design},
  volume={115},
  pages={181--193},
  year={2019},
  publisher={Elsevier}
}

@article{wang2014rigid,
  title={As-rigid-as-possible spherical parametrization},
  author={Wang, Chunxue and Liu, Zheng and Liu, Ligang},
  journal={Graphical models},
  volume={76},
  number={5},
  pages={457--467},
  year={2014},
  publisher={Elsevier}
}

@article{zhao2021s3reg,
  title={S3Reg: superfast spherical surface registration based on deep learning},
  author={Zhao, Fenqiang and Wu, Zhengwang and Wang, Fan and Lin, Weili and Xia, Shunren and Shen, Dinggang and Wang, Li and Li, Gang},
  journal={IEEE transactions on medical imaging},
  volume={40},
  number={8},
  pages={1964--1976},
  year={2021},
  publisher={IEEE}
}

@article{liao2024convergence,
  title={Convergence of Dirichlet Energy Minimization for Spherical Conformal Parameterizations},
  author={Liao, Wei-Hung and Huang, Tsung-Ming and Lin, Wen-Wei and Yueh, Mei-Heng},
  journal={Journal of Scientific Computing},
  volume={98},
  number={1},
  pages={29},
  year={2024},
  publisher={Springer}
}

@article{liu2026spherical,
  title={Spherical Area-Preserving Parameterization via Energy Minimization},
  author={Liu, Shu-Yung and Yueh, Mei-Heng},
  journal={SIAM Journal on Imaging Sciences},
  volume={19},
  number={1},
  pages={207--235},
  year={2026},
  publisher={SIAM}
}

@article{haker2002conformal,
  title={Conformal surface parameterization for texture mapping},
  author={Haker, Steven and Angenent, Sigurd and Tannenbaum, Allen and Kikinis, Ron and Sapiro, Guillermo and Halle, Michael},
  journal={IEEE Transactions on Visualization and Computer Graphics},
  volume={6},
  number={2},
  pages={181--189},
  year={2002},
  publisher={IEEE}
}

@inproceedings{angenent1999conformal,
  title={Conformal geometry and brain flattening},
  author={Angenent, Sigurd and Haker, Steven and Tannenbaum, Allen and Kikinis, Ron},
  booktitle={International Conference on Medical Image Computing and Computer-Assisted Intervention},
  pages={271--278},
  year={1999},
  organization={Springer}
}

@article{angenent2002laplace,
  title={On the Laplace-Beltrami operator and brain surface flattening},
  author={Angenent, Sigurd and Haker, Steven and Tannenbaum, Allen and Kikinis, Ron},
  journal={IEEE transactions on medical imaging},
  volume={18},
  number={8},
  pages={700--711},
  year={2002},
  publisher={IEEE}
}

@article{yeo2009spherical,
  title={Spherical demons: fast diffeomorphic landmark-free surface registration},
  author={Yeo, BT Thomas and Sabuncu, Mert R and Vercauteren, Tom and Ayache, Nicholas and Fischl, Bruce and Golland, Polina},
  journal={IEEE transactions on medical imaging},
  volume={29},
  number={3},
  pages={650--668},
  year={2009},
  publisher={IEEE}
}

@article{wang2016bijective,
  title={Bijective spherical parametrization with low distortion},
  author={Wang, Chunxue and Hu, Xin and Fu, Xiaoming and Liu, Ligang},
  journal={Computers \& Graphics},
  volume={58},
  pages={161--171},
  year={2016},
  publisher={Elsevier}
}

@article{meng2016tempo,
  title={Tempo: feature-endowed Teichmuller extremal mappings of point clouds},
  author={Meng, Ting Wei and Choi, Gary Pui-Tung and Lui, Lok Ming},
  journal={SIAM Journal on Imaging Sciences},
  volume={9},
  number={4},
  pages={1922--1962},
  year={2016},
  publisher={SIAM}
}

@article{lui2013texture,
  title={Texture map and video compression using Beltrami representation},
  author={Lui, Lok Ming and Lam, Ka Chun and Wong, Tsz Wai and Gu, Xianfeng},
  journal={SIAM Journal on Imaging Sciences},
  volume={6},
  number={4},
  pages={1880--1902},
  year={2013},
  publisher={SIAM}
}

@article{lui2015splitting,
  title={A splitting method for diffeomorphism optimization problem using Beltrami coefficients},
  author={Lui, Lok Ming and Ng, Tsz Ching},
  journal={Journal of Scientific Computing},
  volume={63},
  number={2},
  pages={573--611},
  year={2015},
  publisher={Springer}
}

@article{lam2014landmark,
  title={Landmark-and intensity-based registration with large deformations via quasi-conformal maps},
  author={Lam, Ka Chun and Lui, Lok Ming},
  journal={SIAM Journal on Imaging Sciences},
  volume={7},
  number={4},
  pages={2364--2392},
  year={2014},
  publisher={SIAM}
}

@article{lui2012optimization,
  title={Optimization of surface registrations using Beltrami holomorphic flow},
  author={Lui, Lok Ming and Wong, Tsz Wai and Zeng, Wei and Gu, Xianfeng and Thompson, Paul M and Chan, Tony F and Yau, Shing-Tung},
  journal={Journal of scientific computing},
  volume={50},
  number={3},
  pages={557--585},
  year={2012},
  publisher={Springer}
}

@article{lui2007landmark,
  title={Landmark constrained genus zero surface conformal mapping and its application to brain mapping research},
  author={Lui, Lok Ming and Wang, Yalin and Chan, Tony F and Thompson, Paul},
  journal={Applied Numerical Mathematics},
  volume={57},
  number={5-7},
  pages={847--858},
  year={2007},
  publisher={Elsevier}
}

@article{lui2010optimized,
  title={Optimized conformal surface registration with shape-based landmark matching},
  author={Lui, Lok Ming and Thiruvenkadam, Sheshadri and Wang, Yalin and Thompson, Paul M and Chan, Tony F},
  journal={SIAM Journal on Imaging Sciences},
  volume={3},
  number={1},
  pages={52--78},
  year={2010},
  publisher={SIAM}
}

@article{xu2025neural,
  title={A neural optimization framework for free-boundary diffeomorphic mapping problems and its applications},
  author={Xu, Zhehao and Lui, Lok Ming},
  journal={arXiv preprint arXiv:2511.11679},
  year={2025}
}

@article{ren2024sugar,
  title={Sugar: Spherical ultrafast graph attention framework for cortical surface registration},
  author={Ren, Jianxun and An, Ning and Zhang, Youjia and Wang, Danyang and Sun, Zhenyu and Lin, Cong and Cui, Weigang and Wang, Weiwei and Zhou, Ying and Zhang, Wei and others},
  journal={Medical Image Analysis},
  volume={94},
  pages={103122},
  year={2024},
  publisher={Elsevier}
}

@article{lamontagne2019oasis,
  title={OASIS-3: longitudinal neuroimaging, clinical, and cognitive dataset for normal aging and Alzheimer disease},
  author={LaMontagne, Pamela J and Benzinger, Tammie LS and Morris, John C and Keefe, Sarah and Hornbeck, Russ and Xiong, Chengjie and Grant, Elizabeth and Hassenstab, Jason and Moulder, Krista and Vlassenko, Andrei G and others},
  journal={medrxiv},
  pages={2019--12},
  year={2019},
  publisher={Cold Spring Harbor Laboratory Press}
}

@article{fischl2012freesurfer,
  title={FreeSurfer},
  author={Fischl, Bruce},
  journal={Neuroimage},
  volume={62},
  number={2},
  pages={774--781},
  year={2012},
  publisher={Elsevier}
}

@article{choi2015flash,
  title={FLASH: Fast landmark aligned spherical harmonic parameterization for genus-0 closed brain surfaces},
  author={Choi, Pui Tung and Lam, Ka Chun and Lui, Lok Ming},
  journal={SIAM Journal on Imaging Sciences},
  volume={8},
  number={1},
  pages={67--94},
  year={2015},
  publisher={SIAM}
}

@article{lui2014teichmuller,
  title={Teichmuller mapping (t-map) and its applications to landmark matching registration},
  author={Lui, Lok Ming and Lam, Ka Chun and Yau, Shing-Tung and Gu, Xianfeng},
  journal={SIAM Journal on Imaging Sciences},
  volume={7},
  number={1},
  pages={391--426},
  year={2014},
  publisher={SIAM}
}

@article{choi2016spherical,
  title={Spherical conformal parameterization of genus-0 point clouds for meshing},
  author={Choi, Gary Pui-Tung and Ho, Kin Tat and Lui, Lok Ming},
  journal={SIAM Journal on Imaging Sciences},
  volume={9},
  number={4},
  pages={1582--1618},
  year={2016},
  publisher={SIAM}
}

@article{lyu2024spherical,
  title={Spherical density-equalizing map for genus-0 closed surfaces},
  author={Lyu, Zhiyuan and Lui, Lok Ming and Choi, Gary PT},
  journal={SIAM Journal on Imaging Sciences},
  volume={17},
  number={4},
  pages={2110--2141},
  year={2024},
  publisher={SIAM}
}

@article{qiu2019computing,
  title={Computing quasi-conformal folds},
  author={Qiu, Di and Lam, Ka-Chun and Lui, Lok-Ming},
  journal={SIAM Journal on Imaging Sciences},
  volume={12},
  number={3},
  pages={1392--1424},
  year={2019},
  publisher={SIAM}
}

@article{lai2014folding,
  title={Folding-free global conformal mapping for genus-0 surfaces by harmonic energy minimization},
  author={Lai, Rongjie and Wen, Zaiwen and Yin, Wotao and Gu, Xianfeng and Lui, Lok Ming},
  journal={Journal of Scientific Computing},
  volume={58},
  pages={705--725},
  year={2014},
  publisher={Springer}
}

@article{gu2004genus,
  title={Genus zero surface conformal mapping and its application to brain surface mapping},
  author={Gu, Xianfeng and Wang, Yalin and Chan, Tony F and Thompson, Paul M and Yau, Shing-Tung},
  journal={IEEE transactions on medical imaging},
  volume={23},
  number={8},
  pages={949--958},
  year={2004},
  publisher={IEEE}
}

@inproceedings{zhao2019spherical,
  title={Spherical U-Net on cortical surfaces: methods and applications},
  author={Zhao, Fenqiang and Xia, Shunren and Wu, Zhengwang and Duan, Dingna and Wang, Li and Lin, Weili and Gilmore, John H and Shen, Dinggang and Li, Gang},
  booktitle={Information Processing in Medical Imaging: 26th International Conference, IPMI 2019, Hong Kong, China, June 2--7, 2019, Proceedings 26},
  pages={855--866},
  year={2019},
  organization={Springer}
}

@article{cohen2018spherical,
  title={Spherical cnns},
  author={Cohen, Taco S and Geiger, Mario and K{\"o}hler, Jonas and Welling, Max},
  journal={arXiv preprint arXiv:1801.10130},
  year={2018}
}

@inproceedings{zhao2022fast,
  title={Fast spherical mapping of cortical surface meshes using deep unsupervised learning},
  author={Zhao, Fenqiang and Wu, Zhengwang and Wang, Li and Lin, Weili and Li, Gang},
  booktitle={International Conference on Medical Image Computing and Computer-Assisted Intervention},
  pages={163--173},
  year={2022},
  organization={Springer}
}

@article{jin2008discrete,
  title={Discrete surface Ricci flow},
  author={Jin, Miao and Kim, Junho and Luo, Feng and Gu, Xianfeng},
  journal={IEEE Transactions on Visualization and Computer Graphics},
  volume={14},
  number={5},
  pages={1030--1043},
  year={2008},
  publisher={IEEE}
}

@article{guo2023automatic,
  title={Automatic landmark detection and registration of brain cortical surfaces via quasi-conformal geometry and convolutional neural networks},
  author={Guo, Yuchen and Chen, Qiguang and Choi, Gary PT and Lui, Lok Ming},
  journal={Computers in Biology and Medicine},
  volume={163},
  pages={107185},
  year={2023},
  publisher={Elsevier}
}

\end{document}